\documentclass[IEEEtran]{article}

\usepackage{amsmath,amsfonts,amssymb,bm}
\usepackage{url,hyperref,breakurl,graphicx,epsfig,color,fullpage,subfigure,fmtcount,algorithmic,algorithm,semtrans, enumerate}
\usepackage{psfrag}
\usepackage[small,bf]{caption}
\usepackage[bottom,hang,flushmargin]{footmisc}

\numberwithin{equation}{section}

\newcommand{\Null}{\text{Null}}

\newcommand{\md}{d}
\newcommand{\corr}{\text{correlation}}
\newcommand{\bDb}{{\bar{\Db}}}

\newcommand{\Rc}{{\cal {R}}}

\newcommand{\cn}{{\operatorname{cone}}}

\definecolor{darkred}{RGB}{150,0,0}
\definecolor{darkgreen}{RGB}{0,150,0}
\definecolor{darkblue}{RGB}{0,0,200}
\hypersetup{colorlinks=true, linkcolor=darkred, citecolor=darkgreen, urlcolor=darkblue}
\definecolor{orange}{RGB}{205, 140,0}
\newcommand{\red}[1]{\textcolor{red}{ #1}}

\newtheorem{theorem}{Theorem}[section]

\newtheorem{lemma}{Lemma}[section]
\newtheorem{corollary}{Corollary}[section]
\newtheorem{proposition}{Proposition}[section]
\newtheorem{definition}{Definition}[section]

\def \endprf{\hfill {\vrule height6pt width6pt depth0pt}\medskip}
\newenvironment{proof}{\noindent {\it Proof.} }{\endprf\par}

\newcommand{\beq}{\begin{equation}}
\newcommand{\eeq}{\end{equation}}
\newcommand{\beqa}{\begin{equation} \begin{aligned}}
\newcommand{\eeqa}{\end{aligned} \end{equation}}
\newcommand{\beqas}{\begin{equation*} \begin{aligned}}
\newcommand{\eeqas}{\end{aligned} \end{equation*}}
\newcommand{\nn}{\nonumber}

\newcommand{\bmat}{\begin{bmatrix}}
\newcommand{\emat}{\end{bmatrix}}

\newcommand{\Meb}{{\bf{A}}}

\newcommand{\Mec}{{\cal{A}}}

\newcommand{\Iden}{{\bf{I}}}

\newcommand{\tr}[1]{\operatorname{tr}\left(#1\right)}
\newcommand{\sgn}[1]{\operatorname{sgn}\left(#1\right)}

\newcommand{\supp}[1]{\operatorname{supp} \left(#1\right)}

\newcommand{\cone}{\operatorname{cone} }
\newcommand{\rank}{\operatorname{rank}}

\newcommand{\iprod}[2]{\left\langle #1 , #2 \right\rangle}

\newcommand{\li}{\left<}
\newcommand{\ri}{\right>}

\newcommand{\E}{\operatorname{\mathbb{E}}}  
\newcommand{\order}[1]{\mathcal{O}\left(#1\right)}
\newcommand{\minimize}{\operatorname{minimize}}

\newcommand{\vc}{\operatorname{vec}}

\newcommand{\norm}[1]{\|#1\|}

\newcommand{\onenorm}[1]{\|#1\|_1}
\newcommand{\twonorm}[1]{\|#1\|_2}

\newcommand{\fronorm}[1]{\|#1\|_{F}}

\newcommand{\nucnorm}[1]{\|#1\|_*}
\newcommand{\abs}[1]{|#1|}

\newcommand{\nor}{\norm{\cdot}}

\newcommand{\Ab}{\mathbf{A}}
\newcommand{\ab}{\mathbf{a}}
\newcommand{\bb}{\mathbf{b}}

\newcommand{\Eb}{\mathbf{E}}

\newcommand{\G}{\mathbf{G}}
\newcommand{\Ib}{\mathbf{I}}

\renewcommand{\U}{\mathbf{U}}
\newcommand{\V}{\mathbf{V}}
\newcommand{\W}{\mathbf{W}}
\newcommand{\X}{\mathbf{X}}
\newcommand{\Y}{\mathbf{Y}}
\newcommand{\Z}{\mathbf{Z}}
\newcommand{\vv}{\mathbf{v}}
\newcommand{\mlow}{m_0}

\newcommand{\w}{\mathbf{w}}

\newcommand{\x}{\mathbf{x}}
\newcommand{\ub}{\mathbf{u}}
\newcommand{\g}{\mathbf{g}}

\newcommand{\vb}{\mathbf{v}}

\newcommand{\Db}{\mathbf{D}}
\newcommand{\gdb}{\mathbf{D}}
\renewcommand{\b}{\mathbf{b}}
\newcommand{\e}{\mathbf{e}}

\newcommand{\y}{\mathbf{y}}
\newcommand{\s}{\mathbf{s}}
\newcommand{\z}{\mathbf{z}}

\renewcommand{\a}{\mathbf{a}}

\newcommand{\h}{\mathbf{h}}

\newcommand{\dil}{{\alpha}}

\newcommand{\A}{\mathcal{A}}
\newcommand{\Ac}{\mathcal{A}}

\newcommand{\Cc}{\mathcal{C}}
\renewcommand{\D}{\mathcal{D}}
\newcommand{\Dc}{\mathcal{D}}
\newcommand{\Gc}{\mathcal{G}}

\newcommand{\Mc}{\mathcal{M}}
\newcommand{\Nc}{\mathcal{N}}

\newcommand{\Pc}{\mathcal{P}}

\newcommand{\Sc}{{\mathcal{S}}}

\newcommand{\Yc}{\mathcal{Y}}

\newcommand{\bx}{{\bf{\bar{x}}}}
\newcommand{\bX}{{\bf{\bar{X}}}}
\newcommand{\Hb}{{\bf{{H}}}}
\newcommand{\R}{\mathbb{R}}

\newcommand{\Sbb}{{\mathbb{S}}}

\newcommand{\Pro}{{\mathbb{P}}}

\newcommand{\tensor}{\tau}
\newcommand{\eps}{\varepsilon}

\newcommand{\p}{{\bf{p}}}

\newcommand{\la}{{\lambda}}

\newcommand{\bSi}{\mathbf{\Sigma}}
\newcommand{\st}{\star}
\newcommand{\pa}{\partial}

\newcommand{\ang}{\rho}
\newcommand{\rsing}{\sigma_{\Cc}}

\newcommand{\subto}{\operatorname{subject~to}}
\newcommand{\minz}{\operatorname{minimize}}

\newcommand{\Cp}{{\Cc^\circ}}

\newcommand{\tcs}{\mathcal{R}}

\newcommand{\best}{ _{\text{\rm best}} }
\newcommand{\comment}[1]{}

\newcommand{\samet}[1]{ {{\color{black}{ #1 }}}}

\title{Simultaneously Structured Models with Application to\\ Sparse and Low-rank Matrices}
\author{Samet Oymak$^{\textit{c}}$, Amin Jalali$^{\textit{w}}$,\\
Maryam Fazel$^{\textit{w}}$, Yonina C. Eldar$^{\textit{t}}$, Babak Hassibi$^{\textit{c}}$\\\\$^{\textit{c}}$ California Institute of Technology\\$^{\textit{w}}$ University of Washington\\$^{\textit{t}}$ Technion}

\begin{document}
\maketitle
\vspace{-0.3in}
\begin{abstract}
The topic of recovery of a structured model given a small number of linear observations has been well-studied in recent years. Examples include recovering sparse or group-sparse vectors, low-rank matrices, and the sum of sparse and low-rank matrices, among others. In various applications in signal processing and machine learning, the model of interest is known to be structured in \emph{several} ways at the same time, for example, a matrix that is simultaneously sparse and low-rank.

Often norms that promote each individual structure are known, and allow for recovery using an order-wise optimal number of measurements (e.g., $\ell_1$ norm for sparsity, nuclear norm for matrix rank). Hence, it is reasonable to minimize a combination of such norms. We show that, surprisingly, if we use multi-objective optimization with these norms, then we can do no better, order-wise, than an algorithm that exploits only one of the present structures. This result suggests that to fully exploit the multiple structures, we need an entirely new convex relaxation, i.e. not one that is a function of the convex relaxations used for each structure. We then specialize our results to the case of sparse and low-rank matrices. We show that a nonconvex formulation of the problem can recover the model from very few measurements, which is on the order of the degrees of freedom of the matrix, whereas the convex problem obtained from a combination of the $\ell_1$ and nuclear norms requires many more measurements. This proves an order-wise gap between the performance of the convex and nonconvex recovery problems in this case. Our framework applies to arbitrary structure-inducing norms as well as to a wide range of measurement ensembles. This allows us to give performance bounds for problems such as sparse phase retrieval and low-rank tensor completion.
\end{abstract}


\section{Introduction}\label{intro}

Recovery of a structured model (signal) given a small number of linear observations has been the focus of many studies recently. Examples include recovering sparse or group-sparse vectors (which gave rise to the area of compressed sensing)~\cite{candes-tao,donoho,candes-tao2}, low-rank matrices~\cite{RFP,CandesRecht-completion}, and the sum of sparse and low-rank matrices~\cite{ChandraParriloWillsky-SL,candes_wright}, among others.
More generally, the recovery of a signal that can be expressed as the sum of a few atoms out of an appropriate atomic set has been studied in~\cite{chandra}.
Canonical questions in this area include: How many generic linear measurements are enough to recover the model by any means? How many measurements are enough for a tractable approach, e.g., solving a convex optimization problem?
In the statistics literature, these questions are posed in terms of sample complexity and error rates for estimators minimizing the sum of a quadratic loss function and a regularizer that reflects the desired structure~\cite{AgarwalNegahbanWainwright2012}.

There are many applications where the model of interest is known to have \emph{several} structures
 at the same time (Section ~\ref{sec:apps}). We then seek a signal that lies in the intersection of several sets defining the individual structures (in a sense that we will make precise later).
The most common convex regularizer (penalty) used to promote all structures together is a linear combination of well-known regularizers for each structure.
However, there is currently no general analysis and understanding of how well such regularization performs in terms of the number of observations required for successful recovery of the desired model.
This paper addresses this ubiquitous yet unexplored problem; i.e., the recovery of
\emph{simultaneously structured models}.

An example of a simultaneously structured model is a matrix that is \emph{simultaneously sparse and low-rank}. One would like to come up with algorithms that exploit both types of structures to minimize the number of measurements required for recovery.
An $n\times n$ matrix with rank $r\ll n$ can be described by $\order{rn}$ parameters, and can be recovered using $\order{rn}$ generic measurements via nuclear norm minimization~\cite{RFP,CandesPlan}.
On the other hand, a block-sparse matrix with a $k \times k$ nonzero block where $k \ll n$ can be described by $k^2$ parameters and can be recovered with $\order{k^2 \log \frac{n}{k}}$ generic measurements using $\ell_1$ minimization.
However, a matrix that is \emph{both} rank $r$ and block-sparse can be described by $\order{rk}$ parameters. The question is whether we can exploit this joint structure to efficiently recover such a matrix with $\order{rk}$ measurements.

In this paper we give a negative answer to this question in the following sense:
if we use multi-objective optimization with the $\ell_1$ and nuclear norms (used for sparse signals and low rank matrices, respectively), then the number of measurements required is lower bounded by $\order{\min\{k^2,rn\}}$. In other words, we need at least this number of observations for the desired signal to lie on the Pareto optimal front traced by the $\ell_1$ norm and the nuclear norm. This means we can do \emph{no better than an algorithm that exploits only one of the two structures}.

We introduce a framework to express general simultaneous structures, and as our main result, we prove that the same phenomenon happens for a general set of structures. 
We are able to analyze a wide range of measurement ensembles, including subsampled standard basis (i.e. matrix completion), Gaussian and subgaussian measurements, and quadratic measurements.
Table~\ref{tab:related} summarizes known results on recovery of some common structured models, along with a result of this paper specialized to the problem of low-rank and sparse matrix recovery. The first column gives the number of parameters needed to describe the model (often referred to as its `degrees of freedom'), the second and third columns show how many generic measurements are needed for successful recovery.
In `nonconvex recovery', we assume we are able to find the global minimum of a nonconvex problem. This is clearly intractable in general, and not a practical recovery method---we consider it as a benchmark for theoretical comparison with the (tractable) convex relaxation in order to determine how powerful the relaxation is.

The first and second rows are the results on $k$ sparse vectors in $\R^n$ and rank $r$ matrices in $\R^{n\times n}$ respectively, \cite{candes_tao,CandesPlan}. The third row considers the recovery of ``low-rank plus sparse'' matrices. Consider a matrix $\X\in\R^{n\times n}$ that can be decomposed as $\X=\X_L+\X_S$ where $\X_L$ is a rank $r$ matrix and $\X_S$ is a matrix with only $k$ nonzero entries. The degrees of freedom of $\X$ is $\order{rn+k}$.
Minimizing the infimal convolution of $\ell_1$ norm and nuclear norm, i.e., $f(\X)=\min_{\Y}\|\Y\|_\st+\la\|\X-\Y\|_1$ subject to random Gaussian measurements on $\X$, gives a convex approach for recovering $\X$.
It has been shown that under reasonable incoherence assumptions, $\X$ can be recovered from $\order{(rn+k)\log^2 n}$ measurements which is suboptimal only by a logarithmic factor \cite{wright12}. Finally, the last row in Table~\ref{tab:related} shows one of the results in this paper.
Let $\X\in\R^{n\times n}$ be a rank $r$ matrix whose entries are zero outside a $k_1\times k_2$ submatrix. The degrees of freedom of $\X$ is $\order{(k_1+k_2)r}$. We consider both convex and non-convex programs for the recovery of this type of matrices. The nonconvex method involves minimizing the number of nonzero rows, columns and rank of the matrix jointly, as discussed in Section \ref{sub:SLR}.
As shown later, $\order{(k_1+k_2)r\log n}$ measurements suffices for this program to successfully recover the original matrix. The convex method minimizes any convex combination of the individual structure-inducing norms, namely the nuclear norm and the $\ell_{1,2}$ norm of the matrix, which encourage low-rank and column/row-sparse solutions respectively. We show that with high probability this program cannot recover the original matrix with fewer than $\Omega(rn)$ measurements.
In summary, while nonconvex method is slightly suboptimal, the convex method performs poorly as the number of measurements scales with $n$ rather than $k_1+k_2$.

\begin{table}
\begin{center}
    \begin{tabular}{  l | l | l |  l }
    Model 				& Degrees of Freedom 	& Nonconvex recovery		& Convex recovery 			   \\ \hline
    Sparse vectors 		& $k$				& $\order{k}$			&  $\order{k\log \frac{n}{k}}$ 		 \\ \hline
    Low rank matrices 	& $r(2n-r)$ 		& $\order{rn}$	&  $\order{rn}$	 \\ \hline
    Low rank plus sparse	& $\order{rn+k}$ 	& not analyzed		& $\order{(rn+k)\log^2 n}$  	 \\\hline
    Low rank and sparse	&  $\order{r(k_1+k_2)}$		& $\order{r(k_1+k_2)\log n}$	& $\Omega(rn)$ 		 \\
    \end{tabular}
\end{center}
    \caption{Summary of results in recovery of structured signals. This paper shows a gap between the performance of convex and nonconvex recovery programs for simultaneously structured matrices (last row).}
    \label{tab:related}
\end{table}

\subsection{Contributions}

This paper describes a general framework for analyzing the recovery of models that have more than one structure, by combining penalty functions corresponding to each structure. The framework proposed includes special cases that are of interest in their own right, e.g., sparse and low-rank matrix recovery and low-rank tensor completion \cite{RechtTensor,TenSur}. Our contributions can be summarized as follows.

\paragraph{Poor performance of convex relaxations.}
We consider a model with several structures and associated structure-inducing norms. For recovery, we consider a multi-objective optimization problem to minimize the individual norms simultaneously.
Using Pareto optimality, we know that minimizing a weighted sum of the norms and varying the weights traces out all points of the Pareto-optimal front (i.e., the trade-off surface, Section \ref{setup}). We obtain a lower bound on the number of measurements for any convex function combining the individual norms. A sketch of our main result is as follows.

\begin{quote}\emph{Given a model $\x_0$ with $\tau$ simultaneous structures, the number of measurements required for recovery with high probability using any linear combination of the individual norms satisfies the lower bound
\[ m \geq c\, m_{\min} = c \min_{i=1,\ldots,\tau} m_i \]
where $m_i$ is an intrinsic lower bound on the required number of measurements when minimizing the $i$th norm only. The term $c$ depends on the measurement ensemble we are dealing with.
}\end{quote}
For the norms of interest, $m_i$ will be approximately proportional to the degrees of freedom of the $i$th model, as well as the sample complexity of the associated norm. With $m_{\min}$ as the bottleneck, this result indicates that the combination of norms performs no better than using only one of the norms, even though the target model 
has a very small degree of freedom.

\paragraph{Different measurement ensembles.} Our characterization of recovery failure is easy to interpret and deterministic in nature. We show that it can be used to obtain probabilistic failure results for various random measurement ensembles. In particular, our results hold for measurement matrices with i.i.d subgaussian rows, quadratic measurements and matrix completion type measurements.

\paragraph{Understanding the effect of weighting.} We characterize the sample complexity of the multi-objective function as a function of the weights associated with the individual norms. Our upper and lower bounds reveal that the sample complexity of the multi-objective function is related to a certain convex combination of the sample complexities associated with the individual norms. We give formulas for this combination as a function of the weights.

\paragraph{Incorporating general cone constraints.}
In addition, we can incorporate side information on $\x_0$, expressed as convex cone constraints. This additional information helps in recovery; however, quantifying how much the cone constraints help is not trivial. Our analysis explicitly determines the role of the cone constraint: Geometric properties of the cone such as its Gaussian width determines the constant factors in the bound on the number of measurements.

\paragraph{Sparse and Low-rank matrix recovery: illustrating a gap.}
As a special case, we consider the recovery of simultaneously sparse and low-rank matrices and prove that there is a significant gap between the performance of convex and non-convex recovery programs. This gap is surprising when one considers similar results in low-dimensional model recovery discussed above in Table~\ref{tab:related}.


\subsection{Applications} \label{sec:apps}

We survey several applications where simultaneous structures arise,
as well as existing results specific to these applications. These applications all involve models with simultaneous structures, but the measurement model and the norms that matter differ among applications. 

\paragraph{Sparse signal recovery from quadratic measurements.}
Sparsity has long been exploited in signal processing, applied mathematics, statistics and computer science for tasks such as compression, denoising, model selection, image processing and more. Despite the great interest in exploiting sparsity in various applications, most of the work to date has focused on recovering sparse or low rank data from linear measurements.
Recently, the basic sparse recovery problem has been generalized to the case in which the measurements are given by nonlinear transforms of the unknown input, \cite{BE12}. A special case of this more general setting is quadratic compressed sensing \cite{SESS11} in which the goal is to recover a sparse vector $\x$ from quadratic measurements $b_i=\x^T\Ab_i\x$. This problem can be linearized by \emph{lifting}, where we wish to recover a ``low rank and sparse'' matrix $\X=\x\x^T$ subject to measurements $b_i=\li\Ab_i,\X\ri$.

Sparse recovery problems from quadratic measurements arise in a variety of problems in optics. One example is sub-wavelength optical imaging \cite{SEA12,SESS11} in which the goal is to recover a sparse image from its far-field measurements, where due to the laws of physics the relationship between the (clean) measurement and the unknown image is quadratic. In \cite{SESS11} the quadratic relationship is a result of using partially-incoherent light. The quadratic behavior of the measurements in \cite{SEA12} arises from coherent diffractive imaging in which the image is recovered from its intensity pattern. Under an appropriate experimental setup, this problem amounts to reconstruction of a sparse signal from the magnitude of its Fourier transform.

A related and notable problem involving sparse and low-rank matrices is Sparse Principal Component Analysis (SPCA), mentioned in Section \ref{sec:disc}. 

\paragraph{Sparse phase retrieval.}
Quadratic measurements appear in phase retrieval problems, in which a signal is to be recovered from the magnitude of its measurements  $b_i=|\a_i^T \x|$, where each measurement is a linear transform of the input $\x \in \R^n$ and $\a_i$'s are arbitrary, possibly complex-valued
measurement vectors. An important case is when $\a_i^T\x$ is the Fourier Transform and $b_i^2$ is the power spectral density. Phase retrieval is of great interest in many applications such as optical imaging \cite{W63,millane}, crystallography \cite{H93}, and more \cite{H89,GS72,F82}.

The problem becomes linear when $\x$ is \emph{lifted} and we consider the recovery of $\X=\x\x^T$ where each measurement takes the form $b_i^2=\li\a_i\a_i^T,\X\ri$.
In \cite{SESS11}, an algorithm was developed to treat phase retrieval problems with sparse $\x$ based on a semidefinite relaxation, and low-rank matrix recovery combined with a row-sparsity constraint on the resulting matrix. More recent works also proposed the use of semidefinite relaxation together with sparsity constraints for phase retrieval \cite{OYDS12,LiVoroninski2012,JOB12,Vetterli}. An alternative algorithm 
was recently designed in \cite{SBE12} based on a greedy search. In \cite{JOB12}, the authors also consider sparse signal recovery based on combinatorial and probabilistic approaches and give uniqueness results under certain conditions. Stable uniqueness in phase retrieval problems is studied in \cite{EldarRobust}. The results of \cite{CESV12,candes_strohmer} applies to general (non-sparse) signals where in some cases \emph{masked} versions of the signal are required. 

\samet{
\paragraph{Fused lasso.}
Suppose the signal of interest is sparse and its entries vary slowly, i.e., the signal can be approximated by a piecewise constant function.  
To encourage sparsity, one can use the $\ell_1$ norm, and to encourage the piece-wise constant structure,
discrete total variation can be used, defined as
\beq
\|\x\|_{TV}=\sum_{i=1}^{n-1} |\x_{i+1}-\x_i| \,. \nn
\eeq
$\|\cdot\|_{TV}$ is basically the $\ell_1$ norm of the gradient of the vector; and is approximately sparse. The resulting optimization problem is known as the fused-lasso \cite{FusedLasso}, and is given as
\beq
\min_\x\|\x\|_1+\la \|\x\|_{TV}~~~\text{s.t.}~~~\Mec(\x)=\Mec(\x_0). \label{FL}
\eeq
To the best of our knowledge, the sample complexity of fused lasso has not been analyzed from a compressed sensing point of view. However, there is a series of recent work on the total variation minimization, which may lead to analysis of \eqref{FL} in the future \cite{Ward}.

We remark that TV regularization is also used together with the nuclear norm to encourage a low-rank and smooth (i.e., slowly varying entries) solution. This regularization finds applications in imaging and physics \cite{TVNuc1,TVNuc2}.
%
%
%

\paragraph{Low-rank tensors.}
Tensors with small Tucker rank can be seen as a generalization of low-rank matrices \cite{Tucker}. In this setup, the signal of interest is the tensor $\X_0\in\R^{n_1\times\dots\times  n_\tensor}$, and $\X_0$ is low-rank along its unfoldings which are obtained by reshaping $\X_0$ as a matrix with size $n_i\times \frac{n}{n_i}$, where $n=\prod_{i=1}^\tensor n_i$. Denoting the $i$'th unfolding by $\mathcal{U}_i(\X_0)$, a standard approach to estimate $\X_0$ from $\y=\Mec(\X_0)$ is minimizing the weighted nuclear norms of the unfoldings,
\beq
\min_{\X}\;\, \sum_{i=1}^\tensor  \la_i\|\mathcal{U}_i(\X)\|_\st~~~\text{subject to}~~~\y=\Mec(\X_0)\label{SSN}
\eeq
Low-rank tensors have applications in machine learning, physics, computational finance and high dimensional PDE's \cite{TenSur}. \eqref{SSN} has been investigated by several papers \cite{TenVis,RechtTensor}. Closer to us, \cite{SquareDeal} recently showed that the convex relaxation \eqref{SSN} performs poorly compared to information theoretically optimal bounds for Gaussian measurements. Our results can extend those to the more applicable tensor completion setup, where we observe the entries of the tensor.
}

Other applications of simultaneously structured signals include Collaborative Hierarchical Sparse Modeling \cite{SRSE11} where sparsity is considered within the non-zero blocks in a block-sparse vector, and the recovery of hyperspectral images where we aim to recover a simultaneously block sparse and low rank matrix from compressed observations \cite{vander}.

\subsection{Outline of the paper}
The paper is structured as follows. Background and definitions are given in Section \ref{setup}. An overview of the main results is provided in Section \ref{sec:main}. Section \ref{measure} discusses some measurement ensembles for which our results apply. Section \ref{upper bound sec} provides upper bounds for the convex relaxations for the Gaussian measurement ensemble. The proofs of the general results are presented in Section \ref{sec:general}. The proofs for the special case of simultaneously sparse and low-rank matrices are given in Section \ref{sec:SLR}, where we compare corollaries of the general results with the results on non-convex recovery approaches, and illustrate a gap. Numerical simulations in Section \ref{sec:numerical} empirically support the results on sparse and low-rank matrices. Future directions of research and discussion of results are in Section \ref{sec:disc}.

\section{Problem Setup}\label{setup}

We begin by reviewing some basic definitions. Our results will be on structure-inducing norms; examples include the $\ell_1$ norm, the $\ell_{1,2}$ norm, and the nuclear norm. The nuclear norm of a matrix is denoted by $\|\cdot\|_\star$ and is the sum of the singular values of the matrix. $\ell_{1,2}$ norm is the sum of the $\ell_2$ norms of the columns of a matrix. 
minimizing the $\ell_1$ norm encourages sparse solutions, and the $\ell_{1,2}$ norm and nuclear norm encourage column-sparse and low-rank solutions respectively, \cite{RFP,CandesRecht-completion,StojnicBlock,YuanL12}; see section \ref{decomposable} for more detailed discussion of these norms and their subdifferentials.
The Euclidean norm is denoted by $\|\cdot\|_2$, i.e., the $\ell_2$ norm for vectors and the
Frobenius norm $\|\cdot\|_F$ for matrices. Overlines denote normalization, i.e., for a vector $\x$ and a matrix $\X$, $\bx=\frac{\x}{\|\x\|_2}$ and $\bX=\frac{\X}{\|\X\|_F}$. The minimum and maximum singular values of a matrix $\Ab$ are denoted by $\sigma_{\min}(\Ab)$ and $\sigma_{\max}(\Ab)$. The set of $n\times n$ positive semidefinite (PSD) and symmetric matrices are denoted by $\Sbb_+^n$ and $\Sbb^n$ respectively. 
$\text{cone}(S)$ denotes the conic hull of a given set $S$. 
$\Mec(\cdot):\R^{n}\rightarrow \R^m$ is a linear measurement operator if $\Mec(\x)$ is equivalent to the matrix multiplication $\Meb\x$ where $\Meb\in\R^{m\times n}$. If $\x$ is a matrix, $\Mec(\x)$ will be a matrix multiplication with a suitably vectorized $\x$. 
In some of our results, we consider Gaussian measurements, in which case $\Meb$ has independent $\Nc(0,1)$ entries.

For a vector $\x\in\R^n$, $\|\x\|$ denotes a general norm and $\norm{\x}^*=\sup_{\norm{z}\leq1} \iprod{\x}{\z}$ is the corresponding dual norm. A subgradient of the norm $\nor$ at $\x$ is a vector $\g$ for which
$\norm{\z} \geq \norm{\x} + \iprod{\g}{\z-\x}$
holds for any $\z$. The set of all subgradients is called the subdifferential and is denoted by $\pa \norm{\x}$. 
The Lipschitz constant of the norm is defined as
\beq
L=\sup_{\z_1\neq \z_2\in\R^n}\frac{\|\z_1\|-\|\z_2\|}{\|\z_1-\z_2\|_2}.\nn
\eeq

\begin{figure}
\begin{center}
\includegraphics[width=0.25\textwidth]{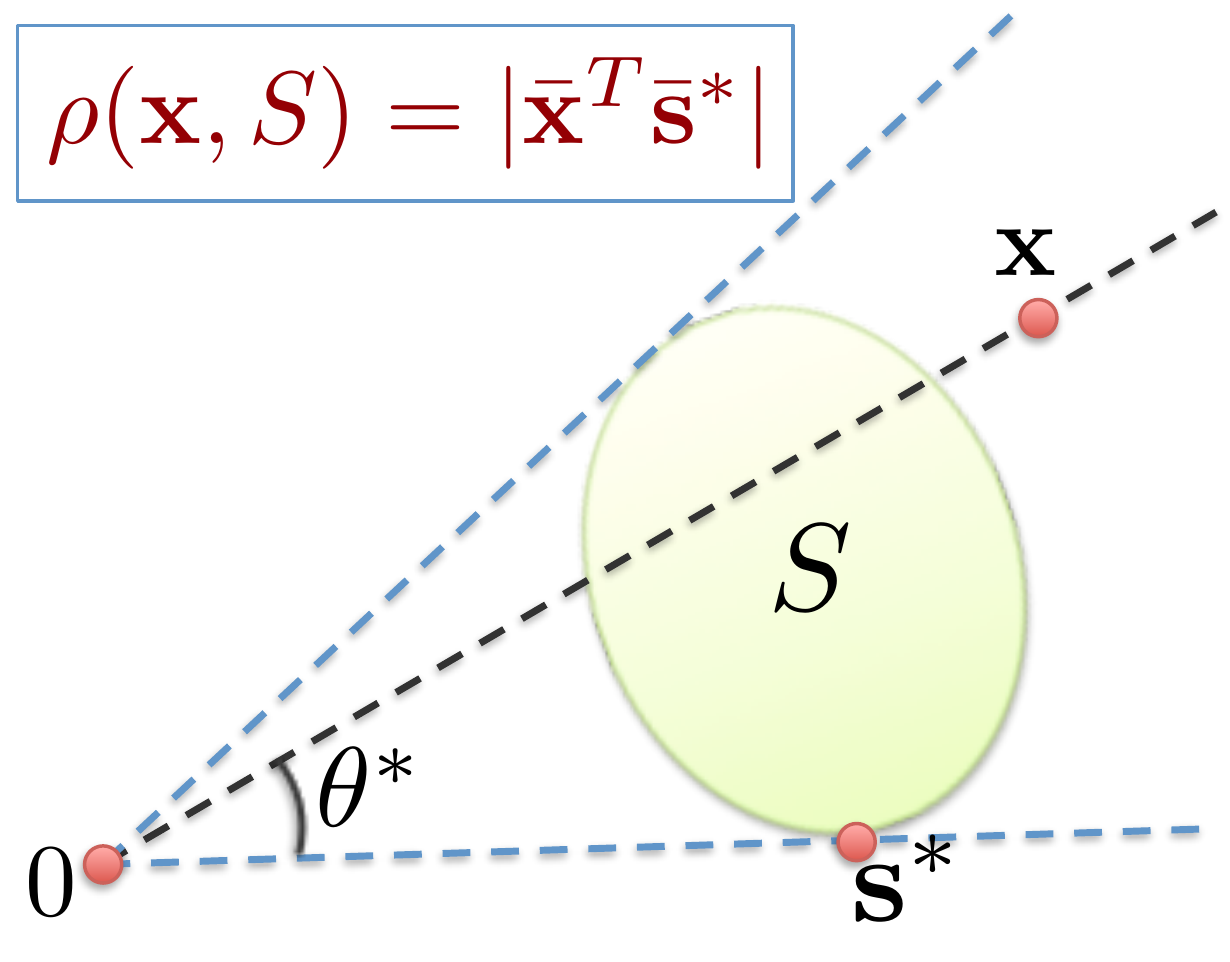}
\end{center}
\caption{Depiction of the $\corr$ between a vector $\x$ and a set $S\,$. $\s^*$ achieves the largest angle with $\x\,$,
hence $\s^*$ has the minimum $\corr$ with $\x$.}
\end{figure}

\begin{definition}[Correlation]\label{def:correlate} Given a nonzero vector $\x$ and a set $S$, $\ang(\x,S)$ is defined as
\beq
\ang(\x,S):=\inf_{0\neq \s\in S} \frac{|\x^T\s|}{\|\x\|_2\|\s\|_2}.\nn
\eeq
\end{definition}

$\ang(\x,S)$ corresponds to the minimum absolute-valued $\corr$ between the vector $\x$ and elements of $S$. Let $\bx=\frac{\x}{\|\x\|_2}$. The $\corr$ between $\x$ and the associated subdifferential has a simple form.
\beq
\ang(\x,\pa \|\x\|)=\inf_{\g\in \pa \|\x\|} \frac{\bx^T\g}{\|\g\|_2}=\frac{\|\bx\|}{\sup_{\g\in \pa \|\x\|}\|\g\|_2} \,.\nn
\eeq
Here, we used the fact that, for norms, subgradients $\g\in \pa\|\x\|$ satisfy $\x^T\g=\|\x\|$, \cite{watson}. The denominator of the right hand side is the local Lipschitz constant of $\|\cdot\|$ at $\x$ and is upper bounded by $L$. Consequently, $\ang(\x,\pa \|\x\|)\geq \frac{\|\bx\|}{L}$. We will denote $\frac{\|\bx\|}{L}$ by $\kappa$. Recently, this quantity has been studied by Mu et al. to analyze the simultaneously structured signals in a similar spirit to us for Gaussian measurements \cite{SquareDeal}\footnote{The work \cite{SquareDeal} is submitted after our initial manuscript; which was projecting the subdifferential onto a carefully chosen subspace to obtain bounds on the sample complexity (see Proposition \ref{bound2}). Inspired from \cite{SquareDeal}, projection onto $\x_0$ and the use of $\kappa$ led to the simplification of the notation and improvement of the results in the current manuscript, in particular, Section \ref{measure}.}. Similar calculations as above gives an alternative interpretation for $\kappa$ which is illustrated in Figure \ref{fig kappa}.

$\kappa$ is a measure of alignment between the vector $\x$ and the subdifferential. For the norms of interest, it is associated with the model complexity. For instance, for a $k$-sparse vector $\x$, $\|\bx\|_1$ lies between $1$ and $\sqrt{k}$ depending on how spiky nonzero entries are. Also $L=\sqrt{n}$. When nonzero entries are $\pm1$, we find $\kappa^2=\frac{k}{n}$. Similarly, given a $\md\times \md$, rank $r$ matrix $\X$, $\|\bX\|_\star$ lies between $1$ and $\sqrt{r}$. If the singular values are spread (i.e. $\pm1$), we find $\kappa^2=\frac{r}{\md}=\frac{r\md}{\md^2}$. In these cases, $\kappa^2$ is proportional to the model complexity normalized by the ambient dimension. 

\begin{figure}[t!]
\begin{center}
\vskip -.3in
\scalebox{1.5}{
\setlength{\unitlength}{0.2cm}
\begin{picture}(32,0)

	\color{black}
	\linethickness{.25mm}
	\put(16,-17){\vector(0,1){14}}	
	\put(9,-10){\vector(1,0){14}}	
	\linethickness{.5mm}
	\put(15.9,-10){\line(1,0){.2}}	
	\put(20.9,-10){\line(1,0){.2}}	
	\put(20.7,-9.3){\tiny{$\x_0$}}	

	\linethickness{.25mm}
	\put(11,-10){\line(1,1){5}}		
	\put(21,-10){\line(-1,1){5}}		
	\put(11,-10){\line(1,-1){5}}		
	\put(21,-10){\line(-1,-1){5}}		

\color{red}\linethickness{.25mm}
	\put(16,-10){\line(1,0){5}}		
\color{darkgreen}\linethickness{.5mm}
	\put(16,-10){\line(1,1){2.5}}	
\color{black}\linethickness{.05mm}
	\put(16,-10){\circle{7.5}}		
\color{black}
\put(18.6,-7){\tiny$\p$}

\end{picture}
}
\vspace{2.1in}
\caption{Consider the scaled norm ball passing through $\x_0\,$, then $\kappa = \frac{\norm{\p}_2}{\norm{\x_0}_2}$,  where $\p$ is any of the closest points on the scaled norm ball to the origin.
}
\label{fig kappa}
\end{center}
\end{figure}
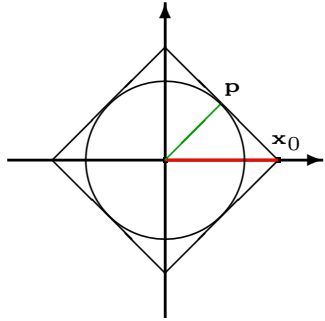

\paragraph{Simultaneously structured models.}
\label{simstr}


We consider a signal $\x_0$ which has several low-dimensional structures $S_1$, $S_2$, \dots, $S_\tau$ (e.g., sparsity, group sparsity, low-rank). Suppose each structure $i$ corresponds to a norm denoted by $\|\cdot\|_{(i)}$ which promotes that structure (e.g., $\ell_1$, $\ell_{1,2}$, nuclear norm). We refer to such an $\x_0$ as a \emph{simultaneously structured model}. 
\subsection{Convex recovery program} \label{recoversec}

We investigate the recovery of the simultaneously structured $\x_0$ from its linear measurements $\Mec(\x_0)$. To recover $\x_0$, we would like to simultaneously minimize the norms $\|\cdot\|_{(i)}$, $i=1,\ldots,\tau$, which leads to a multi-objective (vector-valued) optimization problem. For all feasible points $\x$ satisfying $\Mec(\x)=\Mec(\x_0)$ and side information $\x\in\Cc$, consider the set of achievable norms $\{ \|\x\|_{(i)} \}_{i=1}^\tau$ denoted as points in $\R^\tau$. The minimal points of this set with respect to the positive orthant $\R^\tau_+$ form the \emph{Pareto-optimal} front, as illustrated in Figure \ref{fig:pareto}.
Since the problem is convex, one can alternatively consider the set 
\[
\{\vb\in\R^\tau:\; \exists \x\in\R^n \; \text{s.t.}\; \x\in\Cc,\, \Ac(\x)=\Ac(\x_0),\, v_i\geq \norm{\x}_{(i)},\; \text{for } i=1,\ldots,\tau  \}, 
\]
which is convex and has the same Pareto optimal points as the original set (see, e.g., \cite[Chapter~4]{boyd}).

\begin{definition}[Recoverability]\label{rec_ability}
We call $\x_0$ recoverable if it is a Pareto optimal point; i.e., there does not exist a feasible $\x' \neq \x$ satisfying $\Mec(\x')=\Mec(\x_0)$ and $\x'\in\Cc$,
with $\norm{\x'}_{(i)} \leq \norm{\x_0}_{(i)}$ for $i=1,\ldots,\tau\,$.
\end{definition}

The vector-valued convex recovery program can be turned into a scalar optimization problem as
\beq
\begin{array}{ll}\label{rec_class}
\underset{\x\in \Cc}{\mbox{minimize}} &  f(\x)= h(\|\x\|_{(1)},\dots,\|\x\|_{(\tau)})\\
\subto	& \Mec(\x)=\Mec(\x_0),
\end{array}
\eeq
where $h:\R_+^\tau \to \R_+$ is convex and non-decreasing in each argument 
(i.e., non-decreasing and strictly increasing in at least one coordinate). For convex problems with strong duality,
it is known that we can recover all of the Pareto optimal points by optimizing weighted sums $f(\x)=\sum_{i=1}^\tau \lambda_i \|\x\|_{(i)}\,$, with positive weights $\lambda_i\,$, among all possible functions $f(\x)= h(\|\x\|_{(1)},\dots,\|\x\|_{(\tau)})\,$. For each $\x_0$ on the Pareto, the coefficients of such recovering function are given by the hyperplane supporting the Pareto at $\x_0$ \cite[Chapter~4]{boyd}.
%

\begin{figure}[t!]
\setlength{\unitlength}{0.8cm}
\vskip .8in
\begin{picture}(16,4)
\linethickness{0.075mm}
	\put(7.5,1){\vector(1,0){6}}
	\put(13.8,.9){$\|\cdot\|_{(1)}$}
	\put(8,.5){\vector(0,1){6}}
	\put(6.4,6){$\|\cdot\|_{(2)}$}
	
	\put(10.15,3.5){\circle*{.1}}
	\qbezier(10.2,6.5)(7.3,.4)(13,2.4)	
		\put(13.1,2.3){$\underline{m}$}
	\qbezier(10.5,6.5)(8.3,2.5)(13,3)	
		\put(13.1,2.9){$m$}
	\qbezier(11,6.5)(8.1,2.25)(13,3.4)	
\end{picture}
\caption{
Suppose $\x_0$ corresponds to the point shown with a dot. We need at least $m$ measurements for $\x_0$ to be recoverable since for any $\underline{m}< m$ this point is not on the Pareto optimal front.}\label{fig:pareto}
\end{figure}
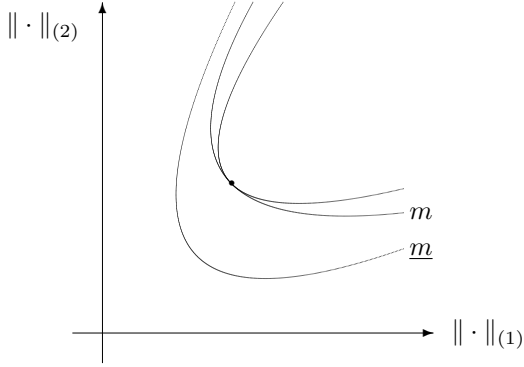

In Figure \ref{fig:pareto}, consider the smallest $m$ that makes $\x_0$ recoverable.
Then one can choose a function $h$ and recover $\x_0$ by \eqref{rec_class} using the $m$ measurements. If the number of measurements is any less, then \emph{no} function can recover $\x_0$. Our goal is to provide lower bounds on $m$.

Note that in \cite{chandra}, Chandrasekaran et al.\ propose a general theory for constructing a suitable penalty, called an \emph{atomic norm}, given a single set of atoms that describes the structure of the target object. In the case of simultaneous structures, this construction requires defining new atoms, and then ensuring the resulting atomic norm can be minimized in a computationally tractable way, which is nontrivial and often intractable. We briefly discuss such constructions as a future research direction in Section \ref{sec:disc}.

\section{Main Results: Theorem Statements} \label{sec:main}
In this section, we state our main theorems that aim to characterize the number of measurements needed to recover a simultaneously structured signal by convex or nonconvex programs. We first present our general results, followed by results for simultaneously sparse and low-rank matrices as a specific but important instance of the general case. The proofs are given in Sections \ref{sec:general} and \ref{sec:SLR}. All of our statements will implicitly assume $\x_0\neq 0$. This will ensure that $\x_0$ is not a trivial minimizer and $0$ is not in the subdifferentials.

\subsection{General simultaneously structured signals}\label{generalres}
This section deals with the recovery of a signal $\x_0$ that is simultaneously structured with $S_1,S_2,\dots,S_\tau$ as described in 
Section \ref{simstr}.
We give a lower bound on the required number of measurements, using the geometric properties of the individual norms.
%

\samet{
\begin{theorem}[Deterministic failure] \label{bound}
 Suppose $\Cc=\R^n$ and,
\beq	\label{lowbb}
\ang(\x_0,\pa f(\x_0)):=\inf_{\g\in\pa f(\x_0)} |{\bar{\g}}^T\bx_0|>\frac{\|\Meb\bx_0\|_2}{\sigma_{\min}(\Meb^T)}.
\eeq
Then, $\x_0$ is not a minimizer of \eqref{rec_class}.

\end{theorem}

Theorem \ref{bound} is deterministic in nature. However, it can be easily specialized to specific random measurement ensembles. The left hand side of \eqref{lowbb} depends only on the vector $\x_0$ and the subdifferential $\pa f(\x_0)$, hence it is independent of the measurement matrix $\Meb$. For simultaneously structured models, we will argue that, the left hand side cannot be made too small, as the subgradients are \emph{aligned} with the signal. On the other hand, the right hand side depends only on $\Meb$ and $\x_0$ and is independent of the subdifferential. In linear inverse problems, $\Meb$ is often assumed to be random. For large class of random matrices, we will argue that, the right hand side is approximately $\sim\sqrt{\frac{m}{n}}$ which will yield a lower bound on the number of required measurements.

Typical measurement ensembles include the following,
\begin{itemize}
\item {\bf{Sampling entries:}} In low-rank matrix and tensor completion problems, we observe the entries of $\x_0$ uniformly at random. In this case, rows of $\Meb$ are chosen from the standard basis in $\R^n$. We should remark that, instead of the standard basis, one can consider other orthonormal bases such as the Fourier basis.
\item {\bf{Matrices with i.i.d. rows:}} $\Meb$ has independent and identically distributed rows with certain moment conditions. This is a widely used setup in compressed sensing as each measurement we make is associated with the corresponding row of $\Meb$ \cite{PlanRIPless}.

\item {\bf{Quadratic measurements:}} Arises in the phase retrieval problem as discussed in Section \ref{sec:apps}.
\end{itemize}

In Section \ref{measure}, we find upper bounds on the right hand side of \eqref{lowbb} for these ensembles. As it will be discussed in Section \ref{measure}, we can do modifications in the rows of $\Meb$ to get better bounds as long as it does not affect its null space. For instance, one can discard the identical rows to improve conditioning. However, as $m$ increases and $\Meb$ has more linearly independent rows, $\sigma_{\min}(\Meb^T)$ will naturally decrease and \eqref{lowbb} will no longer hold after a certain point. In particular, \eqref{lowbb} cannot hold beyond $m\geq n$ as $\sigma_{\min}(\Meb^T)=0$. This is indeed natural as the system becomes overdetermined. 

The following proposition lower bounds the left hand side of \eqref{lowbb} in an interpretable manner. In particular, the $\corr$ $\ang(\x_0,\pa f(\x_0))$ can be lower bounded by the smallest individual $\corr$.

\begin{proposition}\label{low subgrad} Let $L_i$ be the Lipschitz constant of the $i$'th norm and $\kappa_i=\frac{\|\bx_0\|_{(i)}}{L_i}$ for $1\leq i\leq \tau$. Set $\kappa_{\min}=\displaystyle\min \{ \kappa_i : \; i=1,\ldots, \tau \}$. We have the following,
\begin{itemize}
\item All functions $f(\cdot)$ in \eqref{rec_class} satisfy, $\ang(\x_0,\pa f(\x_0))\geq \kappa_{\min}$\footnote{The lower bound $\kappa_{\min}$ is directly comparable to Theorem $5$ of \cite{SquareDeal}. Indeed, our lower bounds on the sample complexity will have the form $\order{\kappa_{\min}^2n}$.}.
\item Suppose $f(\cdot)$ is a weighted linear combination $f(\x)=\sum_{i=1}^\tau \la_i \|\x\|_{(i)}$ for nonnegative $\{\la_i\}_{i=1}^\tau$. Let $\bar{\la}_i=\frac{\la_iL_i}{\sum_{i=1}^\tau \la_iL_i}$ for $1\leq i\leq \tau$. Then, $\ang(\x_0,\pa f(\x_0))\geq  \sum_{i=1}^\tau \bar{\la}_i \kappa_i$.

\end{itemize}
\end{proposition}
\begin{proof} From Lemma \ref{subdiff_chain}, any subgradient of $f(\cdot)$ can be written as, $\g=\sum_{i=1}^\tau w_i\g_i$ for some nonnegative $w_i$'s. On the other hand, from \cite{watson}, $\li\bx_0,\g_i\ri=\|\bx_0\|_{(i)}$. Combining, we find,
\beq
\g^T\bx_0=\sum_{i=1}^\tau w_i\|\bx_0\|_{(i)}.\nn
\eeq
From triangle inequality, $\|\g\|_2\leq \sum_{i=1}^\tau w_iL_i$. To conclude, we use, 
\beq
\frac{\sum_{i=1}^\tau w_i\|\bx_0\|_{(i)}}{\sum_{i=1}^\tau w_iL_i}\geq\min_{1\leq i\leq \tau}\frac{w_i\|\bx_0\|_{(i)}}{w_iL_i}=\kappa_{\min}.\label{subs bar la}
\eeq
For the second part, we use the fact that for the weighted sums of norms, $w_i=\la_i$ and subgradients has the form $\g=\sum_{i=1}^\tau \la_i\g_i$, \cite{boyd}. Then, substitute $\bar{\la}_i$ for $\la_i$ on the left hand side of \eqref{subs bar la}.

\end{proof}

Before stating the next result, let us give a relevant definition regarding the average distance between a set and a random vector.
\begin{definition}[Gaussian distance]\label{cwidth}
Let $\Mc$ be a closed convex set in $\R^n$ and let $\h\in\R^n$ be a vector with independent standard normal entries. Then, the Gaussian distance of $\Mc$ is defined as
\beq
\Db(\Mc)={\E[\inf_{\vb\in \Mc}\|\h-\vb\|_2]}\nn
\eeq
\end{definition}

When $\Mc$ is a cone, we have $0\leq \Db(\Mc)\leq \sqrt{n}$. Similar definitions have been used extensively in the literature, such as Gaussian width \cite{chandra}, statistical dimension \cite{ALMT} and mean width \cite{VerEst}. For notational simplicity, let the normalized distance be $\bDb(\Mc)=\frac{\Db(\Mc)}{\sqrt{n}}$.

We will now state our result for Gaussian measurements; which can additionally include cone constraints for the lower bound. One can obtain results for the other ensembles by referring to Section \ref{measure}.

\begin{theorem}[Gaussian lower bound]\label{main31}
Suppose $\Meb$ has independent $\Nc(0,1)$ entries. Whenever $m\leq m_{low}$, $\x_0$ will not be a minimizer of any of the recovery programs in \eqref{rec_class}
with probability at least $1-10\exp(-\frac{1}{16}\min\{m_{low},(1-\bDb(\Cc))^2n\})$, where
\[
 m_{low}  \;\triangleq\; \frac{(1-\bDb(\Cc))n\kappa_{\min}^2}{100}.
\]

\end{theorem}

\noindent{\bf{Remark:}} When $\Cc=\R^n$, $\bDb(\Cc)=0$ hence, the lower bound simplifies to $m_{low}=\frac{n\kappa_{\min}^2}{100}$.

Here $\bDb(\Cc)$ depends only on $\Cc$ and can be viewed as a constant. For instance, for the positive semidefinite cone, we show $\bDb(\Sbb^n_+)< \frac{\sqrt{3}}{2}$. Observe that for a smaller cone $\Cc$, it is reasonable to expect a smaller lower bound to the required number of measurements. Indeed, as $\Cc$ gets smaller, $\Db(\Cc)$ gets larger. 


As discussed before, there are various options for the scalarizing function in \eqref{rec_class}, with one choice being the weighted sum of norms. In fact, for a recoverable point $\x_0$ there always exists a weighted sum of norms which recovers it.
This function is also often the choice in applications, where the space of positive weights is searched for a good combination. Thus, we can state the following theorem as a general result.


%

\begin{corollary}[Weighted lower bound]\label{maincor3}
Suppose $\Meb$ has i.i.d $\Nc(0,1)$ entries and $f(\x)=\sum_{i=1}^\tau \la_i \|\x\|_{(i)}$ for nonnegative weights $\{\la_i\}_{i=1}^\tau$. Whenever $m\leq m_{low}'$, $\x_0$ will not be a minimizer of the recovery program \eqref{rec_class}
with probability at least $1-10\exp(-\frac{1}{16}\min\{m_{low}',(1-\bDb(\Cc))^2n\})$, where
\[
 m_{low}'  \;\triangleq\; \frac{n (1-\bDb(\Cc)) (\sum_{i=1}^\tau \bar{\la}_i \kappa_i)^2}{ 100},
\]
and $\bar{\la}_i=\frac{\la_iL_i}{\sum_{i=1}^\tau \la_iL_i}$.
\end{corollary}
}
Observe that Theorem \ref{main31} is stronger than stating ``a particular function $h(\|\x\|_{(1)},\dots,\|\x\|_{(\tau)})$ will not work''. Instead, our result states that with high probability none of the programs in the class \eqref{rec_class} can return $\x_0$ as the optimal unless the number of measurements are sufficiently large.

To understand the result better, note that the required number of measurements is proportional to  $\kappa_{min}^2n$ which is often proportional to the sample complexity of the best individual norm. As we have argued in Section \ref{simstr}, $\kappa_i^2n$ corresponds to how structured the signal is. For sparse signals it is equal to the sparsity, and for a rank $r$ matrix, it is equal to the degrees of freedom of the set of rank $r$ matrices. Consequently, Theorem \ref{main31} suggests that even if the signal satisfies multiple structures, the required number of measurements is effectively determined by only one dominant structure.

Intuitively, the degrees of freedom of a simultaneously structured signal can be much lower, which is provable for the S\&L matrices. Hence, there is a considerable gap between the expected measurements based on model complexity and the number of measurements needed for recovery via \eqref{rec_class} ($\kappa_{\min}^2n$).


\subsection{Simultaneously Sparse and Low-rank Matrices} \label{sub:SLR}

We now focus on a special case, namely simultaneously sparse and low-rank (S$\&$L) matrices. We consider matrices with nonzero entries contained in a small submatrix where the submatrix itself is low rank.
Here, norms of interest are $\|\cdot\|_{1,2}$, $\|\cdot\|_1$ and $\|\cdot\|_\star$ and the cone of interest is the PSD cone.
We also consider nonconvex approaches and contrast the results with convex approaches. For the nonconvex problem, we replace the norms $\|\cdot\|_1,\|\cdot\|_{1,2},\|\cdot\|_\star$ with the functions $\|\cdot\|_0,\|\cdot\|_{0,2},\text{rank}(\cdot)$ which give the number of nonzero entries, the number of nonzero columns and rank of a matrix respectively and use the same cone constraint as the convex method.
We show that convex methods perform poorly as predicted by the general result in Theorem \ref{main31}, while nonconvex methods require optimal number of measurements (up to a logarithmic factor). Proofs are given in Section \ref{sec:SLR}.

\begin{table}
\begin{center}
    \begin{tabular}{ l| l | l | l |  l}
    Model 		&$f(\cdot)$	& $L$ 		& $\|\bx_0\|\leq $ & $n\kappa^2\leq $			   \\ \hline
    $k$ sparse vector 	& $\|\cdot\|_{1}$	& $\sqrt{n}$				& $\sqrt{k}$			&  $k$	 \\ \hline
        $k$ column-sparse matrix &$\|\cdot\|_{1,2}$&$\sqrt{\md}$ 		& $\sqrt{k}$	& $k\md$	 \\ \hline
 Rank $r$ matrix &	 $\|\cdot\|_{\star}$& $\sqrt{\md}$ 	& $\sqrt{r}$		& $r\md$ 	 \\\hline
 
 S\&L $(k,k,r)$ matrix&	 $h(\|\cdot\|_{\star},\|\cdot\|_1)$& $-$ 	& $-$		& $\min\{k^2,r\md\}$   	 \\\hline
    \end{tabular}
\end{center}
    \caption{Summary of the parameters that are discussed in this section. The last three lines is for a $\md\times \md$ S\&L ($k,k,r$) matrix where $n=\md^2$. In the fourth column, the corresponding entry for S\&L is $\kappa_{\min}=\min\{\kappa_{\ell_1},\kappa_\star\}$.} 
    \label{tab:related}
\end{table}

\begin{definition}
\label{models}
We say $\X_0\in\R^{\md_1\times \md_2}$ is an S$\&$L matrix with $(k_1,k_2,r)$ if the smallest submatrix that contains nonzero entries of $\X_0$ has size $k_1\times k_2$ and $\rank{(\X_0)}=r$. When $\X_0$ is symmetric, let $\md=\md_1=\md_2$ and $k=k_1=k_2$. We consider the following cases.
\begin{itemize}
\item[(a)] {\emph{General:}} $\X_0\in\R^{\md_1\times \md_2}$ is S$\&$L with $(k_1,k_2,r)$.
\item[(b)] {\emph{PSD model:}} $\X_0\in\R^{n\times n}$ is PSD and S$\&$L with $(k,k,r)$.
\end{itemize}
\end{definition}
We are interested in S\&L matrices with $k_1\ll \md_1,k_2\ll \md_2$ so that the matrix is sparse, and $r\ll \min\{k_1,k_2\}$ so that the submatrix containing the nonzero entries is low rank.
Recall from Section \ref{recoversec} that our goal is to recover $\X_0$ from linear observations $\Mec(\X_0)$ via convex or nonconvex optimization programs. The measurements can be equivalently written as $\Meb \vc(\X_0)$, where $\Meb\in\R^{m\times \md_1\md_2}$ and $\vc(\X_0)\in\R^{\md_1 \md_2}$ denotes the vector obtained by stacking the columns of $\X_0$.

Based on the results in Section \ref{generalres}, we obtain lower bounds on the number of measurements for convex recovery. We additionally show that significantly fewer measurements are sufficient for non-convex programs to uniquely recover $\X_0$; thus proving a performance gap between convex and nonconvex approaches. The following theorem summarizes the results.
\begin{theorem}[\textbf{Performance of S\&L matrix recovery}] \label{main_thm_sec3}
Suppose $\Mec(\cdot)$ is an i.i.d Gaussian map and consider recovering $\X_0\in\R^{\md_1\times \md_2}$ via
\beq\label{mainpro}
\underset{\X\in \Cc}{\mbox{minimize}}~f(\X)~~~\mbox{subject to}~~~\Mec(\X)=\Mec(\X_0).
\eeq
For the cases given in Definition \ref{models}, the following convex and nonconvex recovery results hold for some positive constants $c_1,c_2$.
\begin{itemize}
\item[(a)] General model:
\begin{itemize}
\item[(a1)] 
Let $f(\X)=\|\X\|_{1,2}+\la_1\|\X^T\|_{1,2}+\la_2\|\X\|_\star$ where $\la_1,\la_2\geq 0$ and $\Cc=\R^{\md_1\times \md_2}$. Then, \eqref{mainpro} will fail to recover $\X_0$ with probability $1-\exp(-c_1\mlow)$ whenever $m\leq c_2\mlow$ where $\mlow=\min\{\md_1k_2,\md_2k_1,(\md_1+\md_2)r\}$.
\item[(a2)]
Let $f(\X)=\frac{1}{k_2}\|\X\|_{0,2}+\frac{1}{k_1}\|\X^T\|_{0,2}+\frac{1}{r}\rank(\X)$ and $\Cc=\R^{\md_1\times \md_2}$. Then, \eqref{mainpro} will uniquely recover $\X_0$ with probability $1-\exp(-c_1m)$ whenever $m\geq c_2 \max\{(k_1+k_2)r,k_1\log\frac{\md_1}{k_1},k_2\log\frac{\md_2}{k_2}\}$.
\end{itemize}

\item[(b)] PSD with $\ell_{1,2}$:
\begin{itemize}
\item[(b1)]
Let $f(\X)=\|\X\|_{1,2}+\la\|\X\|_\star$ where $\la\geq 0$ and $\Cc=\Sbb_+^{\md}$. Then, \eqref{mainpro} will fail to recover $\X_0$ with probability $1-\exp(-c_1{r\md})$ whenever $m\leq c_2r\md$.
\item[(b2)]
Let $f(\X)=\frac{2}{k}\|\X\|_{0,2}+\frac{1}{r}\rank(\X)$ and $\Cc=\Sbb^{\md}$. Then, \eqref{mainpro} will uniquely recover $\X_0$ with probability $1-\exp(-c_1m)$ whenever $m\geq c_2 \max\{rk,k\log\frac{\md}{k}\}$.
\end{itemize}

\item[(c)] PSD with $\ell_1$:
\begin{itemize}
\item[(c1)]
Let $f(\X)=\|\X\|_{1}+\la\|\X\|_\star$ and $\Cc=\Sbb_+^{\md}$. Then, \eqref{mainpro} will fail to recover $\X_0$ with probability $1-\exp(-c_1\mlow)$ for all possible $\la\geq 0$ whenever $m\leq c_2 \mlow$ where $\mlow=\min\{\|\bX_0\|_1^2,\|\bX_0\|_\star^2\md\}$.
\item[(c2)]
Suppose $\text{rank}(\X_0)=1$. Let $f(\X)=\frac{1}{k^2}\|\X\|_{0}+{\rank(\X)}$ and $\Cc=\Sbb^{\md}$. Then, \eqref{mainpro} will uniquely recover $\X_0$ with probability $1-\exp(-c_1m)$ whenever $m\geq c_2 k\log\frac{\md}{k}$.
\end{itemize}
\end{itemize}
\end{theorem}

\noindent {\bf{Remark on ``PSD with $\ell_1$'':}} In the special case, $\X_0=\ab\ab^T$ for a $k$-sparse vector $\ab$, we have $\mlow=\min\{{\|\bar{\ab}\|_1^4},\md\}$. When nonzero entries of $\ab$ are $\pm1$, we have $\mlow=\min\{k^2,\md\}$. 
\begin{table}
\begin{center}
    \begin{tabular}{  l | l | l }
    Setting				& Nonconvex sufficient $m$	& Convex required $m$ \\ \hline
    General model 		&   $\order{\max\{rk,k\log\frac{\md}{k}\}}$		& $\Omega(r\md)$	 \\ \hline
    PSD with $\ell_{1,2}$ 	& $\order{\max\{rk,k\log\frac{\md}{k}\}}$	& $\Omega(r\md)$		 \\\hline
    PSD with $\ell_1$ &$\order{k\log\frac{\md}{k}}$				&  $\Omega(\min\{k^2,r\md\})$	 \\ \hline		
    \end{tabular}
\end{center}
    \caption{Summary of recovery results for models in Definition \ref{models}, assuming $\md_1=\md_2=\md$ and $k_1=k_2=k$. For the PSD with $\ell_1$ case, we assume $\frac{\|\bX_0\|_1}{k}$ and $\frac{\|\bX_0\|_\star}{\sqrt{r}}$ to be approximately constants for the sake of simplicity. Nonconvex approaches are optimal up to a logarithmic factor, while convex approaches perform poorly.}
    \label{tab:SLR}
\end{table}

The nonconvex programs require almost the same number of measurements as the degrees of freedom (or number of parameters) of the underlying model. For instance, it is known that the degrees of freedom of a rank $r$ matrix of size $k_1\times k_2$ is simply $r(k_1+k_2-r)$ which is $\order{(k_1+k_2)r}$. Hence, the nonconvex results are optimal up to a logarithmic factor. On the other hand, our results on the convex programs that follow from Theorem \ref{main31} show that the required number of measurements are significantly larger. Table \ref{tab:SLR} provides a quick comparison of the results on S\&L.

For the S\&L (k,k,r) model, from standard results one can easily deduce that \cite{RFP,StojnicBlock,candes-tao2},
\begin{itemize}
\item $\ell_{1}$ penalty only: requires at least $k^2$,
\item $\ell_{1,2}$ penalty only: requires at least $k\md$,
\item Nuclear norm penalty only: requires at least $r\md$ measurements.
\end{itemize}
These follow from the model complexity of the sparse, column-sparse and low-rank matrices. Theorem \ref{main31} shows that, combination of norms require at least as much as the best individual norm. For instance, combination of $\ell_{1}$ and the nuclear norm penalization yields the lower bound $\order{\min\{k^2,r\md\}}$ for S\&L matrices whose singular values and nonzero entries are spread. This is indeed what we would expect from the interpretation that $\kappa^2 n$ is often proportional to the sample complexity of the corresponding norm and, the lower bound $\kappa_{\min}^2 n$ is proportional to that of the best individual norm.

As we saw in Section \ref{generalres}, adding a cone constraint to the recovery program does not help in reducing the lower bound by more than a constant factor.
In particular, we discuss the positive semidefiniteness assumption that is beneficial in the sparse phase retrieval problem,, and show that the number of measurements remain high even when we include this extra information.
On the other hand, the nonconvex recovery programs performs well even without the PSD constraint.

We remark that, we could have stated Theorem \ref{main_thm_sec3} for more general measurements given in Section \ref{measure} without the cone constraint. For instance, the following result holds for the weighted linear combination of individual norms and for the subgaussian ensemble.
\begin{corollary} \label{weighted s&l}Suppose $\X_0\in\R^{\md\times \md}$ obeys the general model with $k_1=k_2=k$ and $\Mec$ is a linear subgaussian map as described in Proposition \ref{sub gauss}. Choose $f(\X)=\la_{\ell_1} \|\X\|_1+\la_\star \|\X\|_\star$, where $\la_{\ell_1}=\beta$, $\la_\star=(1-\beta)\sqrt{\md}$ and $0\leq \beta\leq 1$. Then, whenever, $m\leq \min\{m_{low},c_1n\}$, where,
\beq
m_{low}=\frac{(\beta \|{\bar{\X}}_0\|_1+(1-\beta) \|{\bar{\X}}_0\|_\star \sqrt{\md})^2}{2},\nn
\eeq
\eqref{mainpro} fails with probability $1-4\exp(-c_2m_{low})$. Here $c_1,c_2>0$ are constants as described in Proposition \ref{sub gauss}.
\end{corollary}
\noindent {\bf{Remark:}} Choosing $\X_0=\ab\ab^T$ where nonzero entries of $\ab$ are $\pm 1$ yields $\frac{1}{2}(\beta k+(1-\beta) \sqrt{\md})^2$ on the right hand side. An explicit construction of an S\&L matrix with maximal $\|{\bar{\X}}\|_1,\|{\bar{\X}}\|_\star$ is provided in Section \ref{explicit construction}.

This corollary compares well with the upper bound obtained in Corollary \ref{S&L upper} of Section \ref{upper bound sec}. In particular, both the bounds and the penalty parameters match up to logarithmic factors. Hence, together, they sandwich the sample complexity of the combined cost $f(\X)$.

%

\section{Measurement ensembles}\label{measure}

This section will make use of standard results on sub-gaussian random variables and random matrix theory to obtain probabilistic statements. We will explain how one can analyze the right hand side of \eqref{lowbb} for,
\begin{itemize}
\item Matrices with sub-gaussian rows,
\item Subsampled standard basis (in matrix completion),
\item Quadratic measurements arising in phase retrieval.
\end{itemize}

\subsection{Sub-gaussian measurements}

 We first consider the measurement maps with sub-gaussian entries. The following definitions are borrowed from \cite{Vershynin}.
\begin{definition}[Sub-gaussian random variable]\label{sub gauss def} A random variable $x$ is sub-gaussian if there exists a constant $K>0$ such that for all $p\geq 1$,
\beq
(\E|x|^p)^{1/p}\leq K\sqrt{p}.\nn
\eeq
The smallest such $K$ is called the sub-gaussian norm of $x$ and is denoted by $\|x\|_{\Psi_2}$. A sub-exponential random variable $y$ is one for which there exists a constant $K'$ such that, $\Pro(|y|>t)\leq \exp(1-\frac{t}{K'})$. $x$ is sub-gaussian if and only if $x^2$ is sub-exponential.
\end{definition}

\begin{definition} [Isotropic sub-gaussian vector] A random vector $\x\in\R^n$ is sub-gaussian if the one dimensional marginals $\x^T\vb$ are sub-gaussian random variables for all $\vb\in\R^n$. The sub-gaussian norm of $\x$ is defined as,
\beq
\|\x\|_{\Psi_2}=\sup_{\|\vb\|=1}\|\x^T\vb\|_{\Psi_2}\nn
\eeq
$\x$ is also isotropic, if its covariance is equal to identity, i.e. $\E\x\x^T=\Iden_n$.
\end{definition}
\begin{proposition} [Sub-gaussian measurements]\label{sub gauss} Suppose $\Meb$ has i.i.d rows in either of the following forms,
\begin{itemize}
\item a copy of a zero-mean isotropic sub-gaussian vector $\ab\in\R^n$, where $\|\ab\|_2=\sqrt{n}$ almost surely.
\item have independent zero-mean unit variance sub-gaussian entries.
\end{itemize}
 Then, there exists constants $c_1,c_2$ depending only on the sub-gaussian norm of the rows, such that, whenever $m\leq c_1n$, with probability $1-4\exp(-c_2m)$, we have,
\beq
\frac{\|\Meb\bx_0\|_2^2}{\sigma_{\min}^2(\Meb^T)}\leq \frac{2m}{n}\nn
\eeq
\end{proposition}

\begin{proof} Using Theorem 5.58 of \cite{Vershynin}, there exists constants $c,C$ depending only on the sub-gaussian norm of $\ab$ such that for any $t\geq 0$, with probability $1-2\exp(-ct^2)$
\beq
\sigma_{\min}(\Meb^T)\geq \sqrt{n}-C\sqrt{m}-t\nn
\eeq
Choosing $t=C\sqrt{m}$ and $m\leq \frac{n}{100C^2}$ would ensure $\sigma_{\min}(\Meb^T)\geq \frac{4\sqrt{n}}{5}$.

Next, we shall estimate $\|\Meb\bx_0\|_2$. $\|\Meb\bx_0\|_2^2$ is sum of i.i.d. sub-exponential random variables identical to $|\ab^T\bx_0|^2$. Also, $\E[|\ab^T\bx_0|^2]=1$. Hence, Proposition 5.16 of \cite{Vershynin} gives, 
\beq
\Pro(\|\Meb\bx_0\|_2^2\geq m+t)\leq 2\exp(-c'\min\{\frac{t^2}{m},{t}\})\nn
\eeq
Choosing $t=\frac{7m}{25}$, we find that $\Pro(\|\Meb\bx_0\|_2^2\geq\frac{32m}{25})\leq 2\exp(-{c''m})$. Combining the two, we obtain,
\beq
\Pro(\frac{\|\Meb\bx_0\|_2^2}{\sigma_{\min}^2(\Meb^T)}\leq \frac{2m}{n})\geq 1-4\exp(-c'''m)\nn
\eeq

The second statement can be proved in the exact same manner by using Theorem 5.39 of \cite{Vershynin} instead of Theorem 5.58.
\end{proof}

{\bf{Remark:}} While Proposition \ref{sub gauss} assumes $\ab$ has fixed $\ell_2$ norm, this can be ensured by properly normalizing rows of $\Meb$ (assuming they stay sub-gaussian). For instance, if the $\ell_2$ norm of the rows are larger than $c\sqrt{n}$ for a positive constant $c$, normalization will not affect sub-gaussianity. Note that, scaling rows of a matrix do not change its null space.

\subsection{Randomly sampling entries} 
We now consider the scenario where each row of $\Meb$ is chosen from the standard basis uniformly at random. Note that, when $m$ is comparable to $n$, there is a nonnegligible probability that $\Meb$ will have duplicate rows. Theorem \ref{bound} does not take this situation into account which would make $\sigma_{\min}(\Meb^T)=0$. In this case, one can discard the copies as they don't affect the recoverability of $\x_0$. This would get rid of the ill-conditioning, as the new matrix is well-conditioned with the exact same null space as the original, and would correspond to a ``sampling without replacement'' scheme where we ensure each row is different.

Similar to achievability results in matrix completion \cite{CandesRecht-completion}, the following failure result requires true signal to be incoherent with the standard basis, where incoherence is characterized by $\|\bx_0\|_{\infty}$, which lies between $\frac{1}{\sqrt{n}}$ and $1$.

\begin{proposition}[Sampling entries] \label{indepsamp}Let $\{\e_i\}_{i=1}^n$ be the standard basis in $\R^n$ and suppose each row of $\Meb$ is chosen from $\{\e_i\}_{i=1}^n$ uniformly at random. Let $\hat{\Meb}$ be the matrix obtained by removing the duplicate rows in $\Meb$. Then, with probability $1-\exp(-\frac{m}{4n\|\bx_0\|_\infty^2})$, we have,
\beq
\frac{\|\hat{\Meb}\bx_0\|_2^2}{\sigma_{\min}^2(\hat{\Meb})}\leq \frac{2m}{n}\nn
\eeq
\end{proposition}
\begin{proof} Let $\hat{\Meb}$ be the matrix obtained by discarding the rows of $\Meb$ that occur multiple times except one of them. Clearly $\Null(\hat{\Meb})=\Null(\Meb)$ hence they are equivalent for the purpose of recovering $\x_0$. Furthermore, $\sigma_{\min}(\hat{\Meb})=1$. Hence, we are interested in upper bounding $\|\hat{\Meb}\bx_0\|_2$. 

%

Clearly $\|\hat{\Meb}\bx_0\|_2\leq \|\Meb\bx_0\|_2$. Hence, we will bound $\|\Meb\bx_0\|_2^2$ probabilistically. Let $\ab$ be the first row of $\Meb$. $|\ab^T\bx_0|^2$ is a random variable, with mean $\frac{1}{n}$ and is upper bounded by $\|\bx_0\|_{\infty}^2$. Hence, applying the Chernoff Bound would yield,
\beq
\Pro(\|\Meb\bx_0\|_2^2\geq \frac{m}{n}(1+\delta))\leq \exp(-\frac{m\delta^2 }{2(1+\delta)n\|\bx_0\|_{\infty}^2})\nn
\eeq
Setting $\delta=1$, we find that, with probability $1-\exp(-\frac{m }{4n\|\bx_0\|_{\infty}^2})$, we have,
\beq
\frac{\|\hat{\Meb}\bx_0\|_2^2}{\sigma_{\min}(\hat{\Meb})^2}\leq \frac{\|\Meb\bx_0\|_2^2}{\sigma_{\min}(\hat{\Meb})^2}\leq \frac{2m}{n}\nn
\eeq

\end{proof}

A significant application of this result would be for the low-rank tensor completion problem, where we randomly observe some entries of a low-rank tensor and try to reconstruct it. A promising approach for this problem is using the weighted linear combinations of nuclear norms of the unfoldings of the tensor to induce the low-rank tensor structure described in \eqref{SSN}, \cite{RechtTensor,TenSur}. Related work \cite{SquareDeal} shows the poor performance of \eqref{SSN} for the special case of Gaussian measurements. Combination of Theorem \ref{bound} and Proposition \ref{indepsamp} will immediately extend the results of \cite{SquareDeal} to the more applicable tensor completion setup (under proper incoherence conditions that bound $\|\bx_0\|_{\infty}$).

\noindent {\bf{Remark:}} In Propositions \ref{sub gauss} and \ref{indepsamp}, we can make the upper bound for the ratio $\frac{\|\Meb\bx_0\|_2^2}{\sigma_{\min}(\Meb)^2}$ arbitrarily close to $\frac{m}{n}$ by changing the proof parameters. Combined with Proposition \ref{low subgrad}, this would suggest that, failure happens, when $m<n\kappa_{\min}$.

\subsection{Quadratic measurements}


As mentioned in the phase retrieval problem, quadratic measurements $|\vb^T\ab|^2$ of the vector $\ab\in\R^\md$ can be linearized by the change of variable $\ab\rightarrow\X_0=\ab\ab^T$ and using $\V=\vb\vb^T$. The following proposition can be used to obtain a lower bound for such ensembles when combined with Theorem \ref{bound}.

\begin{proposition} \label{quad prop}Suppose we observe quadratic measurements $\Mec(\X_0)\in\R^m$ of a matrix $\X_0=\ab\ab^T\in\R^{\md\times \md}$. Here, assume that $i$'th entry of $\Mec(\X_0)$ is equal to $|\vb_i^T\ab|^2$ where $\{\vb_i\}_{i=1}^m$ are independent vectors, either with $\Nc(0,1)$ entries or are uniformly distributed over the sphere with radius $\sqrt{\md}$. Then, there exists absolute constants $c_1,c_2>0$ such that whenever $m<\frac{c_1\md}{\log \md}$, with probability $1-2e\md^{-2}$,
\beq
\frac{\|\Mec(\bX_0)\|_2}{\sigma_{\min}(\Meb^T)}\leq \frac{c_2\sqrt{m}\log \md}{\md}\nn
\eeq
\end{proposition}

%


\begin{proof} 
Let $\V_i=\vb_i\vb_i^T$. Without loss of generality, assume $\vb_i$'s are uniformly distributed over sphere with radius $\sqrt{\md}$. To lower bound $\sigma_{\min}(\Meb^T)$, we will estimate the coherence of its columns, defined by,
\beq
\mu(\Meb^T)=\max_{i\neq j} \frac{|\li\V_i,\V_j\ri|}{\|\V_i\|_F\|\V_j\|_F}=\frac{(\vb_i^T\vb_j)^2}{\md^2}\nn
\eeq
Section 5.2.5 of \cite{Vershynin} states that sub-gaussian norm of $\vb_i$ is bounded by an absolute constant. Hence, conditioned on $\vb_j$ (which satisfies $\|\vb_j\|_2=\sqrt{\md}$), $\frac{(\vb_i^T\vb_j)^2}{\md}$ is a subexponential random variable with mean $1$. Hence, using Definition \ref{sub gauss def}, there exists a constant $c>0$ such that,
\beq
\Pro(\frac{(\vb_i^T\vb_j)^2}{\md}> c\log \md)\leq e\md^{-4}\nn
\eeq 
Union bounding over all $i,j$ pairs ensure that with probability $e\md^{-2}$ we have $\mu(\Meb^T)\leq c\frac{\log \md}{\md}$. Next, we use the standard result that for a matrix with columns of equal length, $\sigma_{\min}(\Meb^T)\geq \md(1-(m-1)\mu)$. The reader is referred to Proposition 1 of \cite{TroppDict}. Hence, $m\leq \frac{\md}{2c \log \md}$, gives $\sigma_{\min}(\Meb^T)\geq \frac{\md}{2}$.

It remains to upper bound $\|\Mec(\bX_0)\|_2\,$. The $i$'th entry of $\Mec(\bX_0)$ is equal to $|\vb_i^T\bar{\ab}|^2$, hence it is subexponential. Consequently, there exists a constant $c'$ so that each entry is upper bounded by $\frac{c'}{2}\log \md$ with probability $1-e\md^{-3}$. Union bounding, and using $m\leq \md$, we find that $\|\Mec(\bX_0)\|_2\leq \frac{c'}{2}\sqrt{m}\log \md$ with probability $1-e\md^{-2}$. Combining with the $\sigma_{\min}(\Meb^T)$ estimate we can conclude.
\end{proof}

\paragraph{Comparison to existing literature.} Proposition \ref{quad prop} is useful to estimate the performance of the sparse phase retrieval problem, in which $\ab$ is a $k$ sparse vector, and we minimize a combination of the $\ell_1$ norm and the nuclear norm to recover $\X_0$. Combined with Theorem \ref{bound}, Proposition \ref{quad prop} gives that, whenever $m\leq \frac{c_1\md}{\log \md}$ and $\frac{c_2\sqrt{m}\log \md}{\md}\leq \min\{\frac{\|\bX_0\|_1}{\md},\frac{\|\bX_0\|_\star}{\sqrt{\md}}\}$, the recovery fails with high probability. Since $\|\bX_0\|_\star=1$ and $\|\bX_0\|_1=\|\bar{\ab}\|_1^2$, the failure condition reduces to,
\beq
m\leq \frac{c}{\log^2 \md}\min\{\|\bar{\ab}\|_1^4,\md\}.\nn
\eeq
When $\bar{\ab}$ is a $k$-sparse vector with $\pm 1$ entries, in a similar flavor to Theorem \ref{main_thm_sec3}, the right hand side has the form $\frac{c}{\log^2 \md}\min\{k^2,\md\}$.


We should emphasize that the lower bound provided in \cite{LiVoroninski2012} is directly comparable to our results. Authors in \cite{LiVoroninski2012} consider the same problem and give two results: first, if $m\geq \order{\|\bar{\ab}\|_1^2 k \log \md}$ then minimizing $\|\X\|_1+ \lambda \tr{\X}$ for suitable value of $\lambda$ over the set of PSD matrices will exactly recover $\X_0$ with high probability. Secondly, their Theorem 1.3 gives a necessary condition (lower bound) on the number of measurements, under which the recovery program fails to recover $\X_0$ with high probability. In particular, their failure condition is $m\leq \min\{m_0,\frac{\md}{40\log \md}\}$ where $
m_0= \frac{\max(\norm{\bar{\ab}}_1^2 - k/2 , 0)^2}{500 \log^2 \md}$.

First, observe that both results have $m\leq \order{\frac{\md}{\log \md}}$ condition. Focusing on the sparsity requirements, when the nonzero entries are sufficiently diffused (i.e. $\|\ab\|_1^2\approx k$) both results yield $\order{\frac{\|\bar{\ab}\|^4}{\log^2 \md}}$ as a lower bound. On the other hand, if $\|\bar{\ab}\|_1\leq \sqrt{\frac{k}{2}}$, their lower bound disappears while our lower bound still requires $\order{\frac{\|\bar{\ab}\|^4}{\log^2 \md}}$ measurements. $\|\bar{\ab}\|_1\leq \sqrt{\frac{k}{2}}$ can happen as soon as the nonzero entries are rather spiky, i.e. some of the entries are much larger than the rest. In this sense, our bounds are tighter. On the other hand, their lower bound includes the PSD constraint unlike ours.

%
%
%

\subsection{Asymptotic regime} While we discussed two cases in the nonasymptotic setup, we believe significantly more general results can be stated asymptotically ($m,n\rightarrow \infty$). For instance, under finite fourth moment constraint, thanks to Bai-Yin law \cite{Bai-Yin}, asymptotically, the smallest singular value of a matrix with i.i.d. unit variance entries concentrate around $\sqrt{n}-\sqrt{m}$. Similarly, $\|\Meb\bx_0\|_2^2$ is sum of independent variables; hence thanks to the law of large numbers, we will have $\frac{\|\Meb\bx_0\|_2^2}{m}\rightarrow 1$. Together, these yield $\frac{\|\Meb\bx_0\|_2}{\sigma_{\min}(\Meb^T)}\rightarrow \frac{\sqrt{m}}{\sqrt{n}-\sqrt{m}}$.

\samet{
\section{Upper bounds}\label{upper bound sec}

We now state an upper bound on the simultaneous optimization for Gaussian measurement ensemble. Our upper bound will be in terms of distance to the dilated subdifferentials.

To accomplish this, we will make use of the recent theory on the sample complexity of the linear inverse problems. It has been recently shown that \eqref{rec_class} exhibits a phase transition from failure with high probability to success with high probability when the number of Gaussian measurements are around the quantity $m_{PT}= \gdb(\text{cone}(\pa f(\x_0)))^2$ \cite{ALMT,chandra}. This phenomenon was first observed by Donoho and Tanner, who calculated the phase transitions for $\ell_1$ minimization and showed that $\gdb(\text{cone}(\pa \|\x_0\|_1))^2\leq 2k\log\frac{en}{k}$ for a $k$-sparse vector in $\R^n$ \cite{DonTan}. All of these works focus on signals with single structure and do not study properties of a penalty that is a combination of norms. The next theorem relates the phase transition point of the joint optimization \eqref{rec_class} to the individual subdifferentials.

\begin{theorem}\label{up bound thm} Suppose $\Meb$ has i.i.d.~$\Nc(0,1)$ entries and let $f(\x)=\sum_{i=1}^\tau\la_i\|\x\|_{(i)}$. For positive scalars $\{\dil_i\}_{i=1}^\tau$, let ${\bar{\la}}_i=\frac{\la_i \dil_i^{-1}}{\sum_{i=1}^\tau \la_i \dil_i^{-1}}$ and define,
\beq
m_{up}(\{\dil_i\}_{i=1}^\tau):= \left(\sum_i {\bar{\la}}_i\gdb(\dil_i\partial \|\x_0\|_{(i)})\right)^2 \nn\label{up bound}
\eeq
If $m\geq (\sqrt{m_{up}}+t)^2+1$, then program \eqref{rec_class} will succeed with probability $1-2\exp(-\frac{t^2}{2})$.
\end{theorem}

\begin{proof} Fix $\h$ as an i.i.d. standard normal vector. Let $\g_i$ be so that $\alpha_i\g_i$ is closest to $\h$ over $\alpha_i\partial \|\x_0\|_{(i)}$. Let $\gamma=(\sum_i\frac{\la_i}{\dil_i})^{-1}$. Then, we may write,
\begin{align}
\inf_{\g'\in\cn(\partial  f(\x_0))} \|\h-\g'\|_2&\leq \inf_{\g\in\partial  f(\x_0)} \|\h-\gamma\g\|_2\nn\\
&\leq \|\h-\gamma \sum_i\la_i\g_i\|_2\nn\\
&=\|\h-\gamma \sum_i \frac{\la_i}{\dil_i}\dil_i\g_i\|_2=\|\h- \sum_i{\bar{\la}}_i\dil_i\g_i\|_2\nn \\
&\leq \sum_i {\bar{\la}}_i\|\h-\dil_i\g_i\|_2 \nn\\
&=\sum_i{\bar{\la}}_i\inf_{\g_i'\in\pa \|\x_0\|_{(i)}}\|\h-\dil_i\g_i'\|_2\nn
\end{align}

Taking the expectations of both sides and using the definition of $\Db(\cdot)$, we find,
\begin{align}
\Db(\text{cone}(\pa f(\x_0)))\leq  \sum_i{\bar{\la}}_i\gdb(\dil_i\partial \|\x_0\|_{(i)})\,.\nn
\end{align}
Using definition of $\Db(\cdot)$, this gives, $m_{up}\geq \Db(\text{cone}(\pa f(\x_0)))^2$. The result then follows from the fact that, when $m\geq (\Db(\text{cone}(\pa f(\x_0)))+t)^2+1$, recovery succeeds with probability $1-2\exp(-\frac{t^2}{2})$. To see this, first, as discussed in Proposition $3.6$ of \cite{chandra}, $\Db(\text{cone}(\pa f(\x_0)))$ is equal to the Gaussian width of the ``tangent cone intersected with the unit ball'' (see Theorem \ref{ETM} for a definition of Gaussian width). Then, Corollary $3.3$ of \cite{chandra} yields the probabilistic statement.
\end{proof}

For Theorem \ref{up bound thm} to be useful, choices of $\dil_i$ should be made wisely. An obvious choice is letting,
\beq\dil_i^*=\arg\min_{\dil_i\geq 0}\Db(\dil_i\pa \|\x_0\|_{(i)})\label{opt alpha}.\eeq With this choice, our upper bounds can be related to the individual sample complexities, which is equal to $\Db(\text{cone}(\partial \|\x_0\|_{(i)}))^2$. Proposition $1$ of \cite{MF13} shows that, if $\|\cdot\|_{(i)}$ is a \emph{decomposable} norm, then,
\beq
\gdb(\text{cone}(\partial \|\x_0\|_{(i)}))\leq \gdb(\dil_i^*\partial \|\x_0\|_{(i)})\leq \gdb(\cone(\partial \|\x_0\|_{(i)}))+6\nn
\eeq
Decomposability is defined and discussed in detail in Section \ref{decomposable}. In particular, $\ell_1,\ell_{1,2}$ and the nuclear norm are decomposable. With this assumption, our upper bound will suggest that, the sample complexity of the simultaneous optimization is smaller than a certain convex combination of individual sample complexities.
\begin{corollary}\label{up bound cor} Suppose $\Meb$ has i.i.d $\Nc(0,1)$ entries and let $f(\x)=\sum_{i=1}^\tau\la_i\|\x\|_{(i)}$ for decomposable norms $\{\|\cdot\|_{(i)}\}_{i=1}^\tau$. Let $\{\dil_i^*\}_{i=1}^\tau$ be as in \eqref{opt alpha} and assume they are strictly positive. Let ${\bar{\la}}_i^*=\frac{\la_i (\dil_i^*)^{-1}}{\sum_{i=1}^\tau \la_i (\dil_i^*)^{-1}}$ and define,
\beq
\sqrt{m_{up}(\{\dil_i^*\}_{i=1}^\tau)}:= \sum_i {\bar{\la}}_i^*\gdb(\cone(\partial \|\x_0\|_{(i)}))+6 \nn\label{up bound100}
\eeq
If $m\geq (\sqrt{m_{up}}+t)^2+1$, then program \eqref{rec_class} will succeed with probability $1-2\exp(-\frac{t^2}{2})$.
\end{corollary}

Here, we used the fact that $\sum_i{\bar{\la}}_i^*=1$ to take $6$ out of the sum over $i$. We note that Corollaries \ref{maincor3} and \ref{up bound cor} can be related in the case of sparse and low-rank matrices. For norms of interest, roughly speaking,
\begin{itemize}
\item $n\kappa_i^2$ is proportional to the sample complexity $\gdb(\text{cone}(\partial \|\x_0\|_{(i)}))^2$.
\item  $L_i$ is proportional to $\frac{\sqrt{n}}{\dil_i^*}$.
\end{itemize}
Consequently, the sample complexity of \eqref{rec_class} will be upper and lower bounded by similar convex combinations.

\subsection{Upper bounds for the S\&L model}
We will now apply the bound obtained in Theorem \ref{up bound thm} for S\&L matrices. To obtain simple and closed form bounds, we will make use of the existing results in the literature.
\begin{itemize}
\item Table II of \cite{MF13}: If $\x_0\in\R^n$ is a $k$ sparse vector, choosing $\alpha_{\ell_1}=\sqrt{2\log\frac{n}{k}}$, $\Db(\alpha_{\ell_1}\pa\|\x_0\|_1)^2\leq 2k\log\frac{en}{k}$.
\item Table 3 of \cite{OymLas}: ~If $\X_0\in\R^{\md\times \md}$ is a rank $r$ matrix, choosing $\alpha_{\star}=2\sqrt{\md}$, $\Db(\alpha_{\star}\pa\|\X_0\|_\star)^2\leq 6\md r+2\md$.
\end{itemize}

\begin{proposition}\label{S&L upper} Suppose $\Meb$ has i.i.d $\Nc(0,1)$ entries and $\X_0\in\R^{\md\times \md}$ is a rank $r$ matrix whose nonzero entries lie on a $k\times k$ submatrix. For $0\leq \beta\leq 1$, let $f(\X)=\la_{\ell_1}\|\X\|_1+\la_{\star}\|\X\|_\star$ where $\la_{\ell_1}=\beta\sqrt{\log\frac{\md}{k}}$ and $\la_\star=(1-\beta)\sqrt{\md}$. Then, whenever,
\beq
m\geq \left(2\beta k\sqrt{\log\frac{e\md}{k}}+(1-\beta)\sqrt{6\md r+2\md}+t\right)^2+1,\nn
\eeq
$\X_0$ can be recovered via \eqref{rec_class} with probability $1-2\exp(-\frac{t^2}{2})$. 
\end{proposition}
\begin{proof}
To apply Theorem \ref{up bound thm}, we will choose $\alpha_{\ell_1}=\sqrt{4\log\frac{\md}{k}}$ and $\alpha_\star=2\sqrt{\md}$. $\X_0$ is effectively an (at most) $k^2$ sparse vector of size $\md^2$. Hence, $\alpha_{\ell_1}=\sqrt{2\log\frac{\md^2}{k^2}}$ and $\Db(\alpha_{\ell_1}\|\X_0\|_1)^2\leq 4k^2\log\frac{e\md}{k}$.

Now, for the choice of $\alpha_{\star}$, we have, $\Db(\alpha_{\star}\|\X_0\|_\star)^2\leq 6\md r+2\md$. Observe that $\alpha_{\ell_1}^{-1}\la_{\ell_1}=\frac{\beta}{2}$, $\alpha_{\star}^{-1}\la_{\star}=\frac{1-\beta}{2}$ and apply Theorem \ref{up bound thm} to conclude.
\end{proof}

}

\section{General Simultaneously Structured Model Recovery} \label{sec:general}
Recall the setup from Section \ref{setup} where we consider a vector $\x_0\in\R^n$ whose structures are associated with a family of norms $\{\nor_{(i)}\}_{i=1}^\tau$ and $\x_0$ satisfies the cone constraint $\x_0\in\Cc$. 
This section is dedicated to the proofs of theorems in Section \ref{generalres} and additional side results where the goal is to find lower bounds on the required number of measurements to recover $\x_0$. 

The following definitions will be helpful for the rest of our discussion. For a subspace $M$, denote its orthogonal complement by ${M^\perp}$. For a convex set $M$ and a point $\x$, we define the projection operator as
\beq\Pc_M(\x) =  \arg\min_{\ub\in M}    \twonorm{\x-\ub}\, .\nn\eeq
Given a cone $\Cc$, denote its dual cone by $\Cc^*$ and polar cone by $\Cp=-\Cc^*$, where $\Cc^*$ is defined as
\beq
\Cc^*=\{\z\big|\li\z,\vv\ri\geq 0~\text{for all}~\vv\in \Cc\}\,.\nn
\eeq

\subsection{Preliminary Lemmas} \label{sec:opt_thm}
We first show that the objective function $\max_{1\leq i\leq \tau}\frac{\norm{\x}_{(i)}}{\norm{\x_0}_{(i)}}$ can be viewed as the `best' among the functions mentioned in \eqref{rec_class} for recovery of $\x_0$.
\begin{lemma}\label{best opt}
Consider the class of recovery programs in \eqref{rec_class}. If the program
\begin{equation}\label{bestone}
\begin{array}{ll}
\underset{\x\in \Cc}{\mbox{minimize}} & f\best(\x) \triangleq \max_{i=1,\ldots,\tau} \frac{\norm{\x}_{(i)}}{\norm{\x_0}_{(i)}}\\
\mbox{subject to} & \Mec(\x)=\Mec(\x_0)
\end{array}
\end{equation}
fails to recover $\x_0$, then any member of this class will also fail to recover $\x_0$.
\end{lemma}

\begin{proof}
Suppose \eqref{bestone} does not have $\x_0$ as an optimal solution and there exists $\x'$ such that $f\best(\x')\leq f\best(\x_0)$, then
\[
\frac{1}{\norm{\x_0}_{(i)}}\|\x'\|_{(i)}\leq f\best(\x')\leq f\best(\x_0)=1, \quad \mbox{for}~ i=1,\ldots, \tau,
\]
which implies,
\beq\label{beat}
\|\x'\|_{(i)}\leq\|\x_0\|_{(i)}, \quad \mbox{for all}~ i=1,\ldots, \tau.
\eeq
Conversely, given \eqref{beat}, we have $f\best(\x')\leq f\best(\x_0)$ from the definition of $f\best$.

Furthermore, since we assume $h(\cdot)$ in \eqref{rec_class} is non-decreasing in its arguments and increasing in at least one of them, \eqref{beat} implies $f(\x')\leq f(\x_0)$ for any such function $f(\cdot)$. Thus, failure of $f\best(\cdot)$ in recovery of $\x_0$ implies failure of any other function in \eqref{rec_class} in this task.
\end{proof}
The following lemma gives necessary conditions for $\x_0$ to be a minimizer of the problem \eqref{rec_class}.

\begin{lemma} \label{kkt_max}
If $\x_0$ is a minimizer of the program \eqref{rec_class}, then there exist $\vb \in\Cc^*$, $\z$, and $\g\in\pa f(\x_0)$ such that
\[
\g-\vb-\Meb^T\z=0 \quad \mbox{and} \quad \li\x_0,\vb \ri=0.
\]
\end{lemma}
The proof of Lemma \ref{kkt_max} follows from the KKT conditions for \eqref{rec_class} to have $\x_0$ as an optimal solution \cite[Section 4.7]{Bertse}.

The next lemma describes the subdifferential of any general function $f(\x)= h(\|\x\|_{(1)},\dots,\|\x\|_{(\tau)})$ as discussed in Section \ref{recoversec}.
\begin{lemma} \label{subdiff_chain}
For any subgradient of the function $f(\x)= h(\|\x\|_{(1)},\dots,\|\x\|_{(\tau)})$ at $\x\neq 0$ defined by convex function $h(\cdot)$, there exists non-negative constants $w_i$, $i=1,\ldots,\tau$ such that
\begin{eqnarray*}
\g = \sum_{i=1}^\tau w_i \g_i
\end{eqnarray*}
where $\g_i \in \pa \norm{\x_0}_{(i)}\,$.
\end{lemma}
\begin{proof}
Consider the function $N(\x) = \bmat \|\x\|_{(1)}, & \dots, & \|\x\|_{(\tau)} \emat^T$ by which we have $f(\x) = h(N(\x))$. By Theorem 10.49 in \cite{rock_var} we have
\[
\pa f(\x) = \bigcup \left\{ \pa (\y^T N(\x)) :\; \y\in \pa h(N(\x)) \right\} \,
\]
where we used the convexity of $f$ and $h$. Now notice that any $\y\in \pa h(N(\x))$ is a non-negative vector because of the monotonicity assumption on $h(\cdot)$. This implies that any subgradient $\g\in \pa f(\x)$ is in the form of $\pa (\w^T N(\x))$ for some nonnegative vector $\w$. The desired result simply follows because subgradients of conic combination of norms are conic combinations of their subgradients, (see e.g. \cite{rock}).
\end{proof}

Using Lemmas \ref{kkt_max} and \ref{subdiff_chain}, we now provide the proofs of Theorems \ref{bound} and \ref{main31}.

\subsection{Proof of Theorem \ref{bound}}

We prove the more general version of Theorem \ref{bound}, which can take care of the cone constraint and alignment of subgradients over arbitrary subspaces. This will require us to extend the definition of $\corr$ to handle subspaces. For a linear subspace $\Rc\in\R^n$ and a set $S\in\R^n$, we define,
\beq
\ang(\Rc,S)=\inf_{0\neq \s\in S} \frac{\|\Pc_\Rc(\s)\|_2}{\|\s\|_2}.\nn
\eeq

\begin{proposition} \label{bound2}Let $\rsing(\Ab^T)=\inf_{\|\z\|_2=1}\frac{\twonorm{\Pc_\Cc(\Meb^T\z)}}{\twonorm{\Meb^T\z}}$. Let $\cal {R}$ be an arbitrary linear subspace orthogonal to the following cone,
\beq
\{\y\in\R^n\big|\x_0^T\y=0,~\y\in\Cc^*\} \,.\label{bad cone}
\eeq
Suppose,
\beqas	\label{lowb}
\ang(\Rc,\pa f(\x_0)):=\inf_{\g\in\pa f(\x_0)} \frac{\|\Pc_\Rc(\g)\|_2}{\|\g\|_2}>\frac{\sigma_{\max}(\Pc_\Rc(\Meb^T))}{\rsing(\Meb^T)\sigma_{\min}(\Meb^T)} \,.
\eeqas
Then, $\x_0$ is not a minimizer of \eqref{rec_class}.
\end{proposition}
\begin{proof}
Suppose $\x_0$ is a minimizer of \eqref{rec_class}. From Lemma \ref{kkt_max}, there exist a $\g\in \pa f(\x_0)$, $\z\in\R^m$ and $\vb\in\Cc^*$ such that
\beqa \label{exist_g}
\g=\Meb^T\z + \vb
\eeqa
and $\li\x_0,\vb\ri=0\,$. We will first eliminate the contribution of $\vb$ in equation \eqref{exist_g}. Projecting both sides of \eqref{exist_g} onto the subspace $\tcs$ gives,
\beq
\Pc_\Rc(\g)=\Pc_\Rc(\Meb^T\z)=\Pc_\Rc(\Meb^T)\z\label{look here}
\eeq
Taking the $\ell_2$ norms,
\beq
\|\Pc_\Rc(\g)\|_2=\|\Pc_\Rc(\Meb^T)\z\|_2\leq \sigma_{\max}(\Pc_\Rc(\Meb^T))\|\z\|_2\label{nosoln}.
\eeq

Since $\vb\in\Cc^*$, from Lemma \ref{decamp} we have $\Pc_\Cc(-\vb)=\Pc_\Cc(\Meb^T\z-\g)=0$. Using Corollary \ref{wanter},
\beq
\|\g\|_2\geq \|\Pc_\Cc(\Meb^T\z)\|_2\label{lab4}.
\eeq
From the initial assumption, for any $\z\in\R^m$, we have,
\beq
\rsing(\Meb^T)\twonorm{\Meb^T\z} 	\leq 	  \twonorm{\Pc_\Cc(\Meb^T\z)}\label{lab1}
\eeq
Combining \eqref{lab4} and \eqref{lab1} yields $\|\g\|_2\geq \rsing(\Meb^T)\|\Meb^T\z\|_2$.
Further incorporating \eqref{nosoln}, we find,
\beqa \nn\label{z_bounds}
\frac{\|\Pc_\Rc(\g)\|_2}{\sigma_{\max}(\Pc_\Rc(\Meb^T))}  	\leq \twonorm{\z} \leq	 \frac{\|\Meb^T\z\|_2 }{\sigma_{\min}(\Meb^T)}\leq  \frac{ \twonorm{\g} }{\rsing(\Meb^T)\sigma_{\min}(\Meb^T)}\,.
\eeqa

Hence, if $\x_0$ is recoverable, there exists $\g\in\pa f(\x_0)$ satisfying,
\beq
\frac{\|\Pc_\Rc(\g)\|_2 }{ \twonorm{\g} } \leq  \frac{\sigma_{\max}(\Pc_\Rc(\Meb^T))}{\rsing(\Meb^T)\sigma_{\min}(\Meb^T)}\,.\nn
\eeq
\end{proof}

To obtain Theorem \ref{bound}, choose ${\cal{R}}=\text{span}(\{\x_0\})$ and $\Cc=\R^n$. This choice of ${\cal{R}}$ yields $\sigma_{\max}(\Pc_\Rc(\Meb^T))=\|\bx_0\bx_0^T\Meb^T\|_2=\|\Meb\bx_0\|_2$ and $\|\Pc_{\Rc}(\g)\|_2=|\bx_0^T\g|$. Choice of $\Cc=\R^n$ yields $\rsing(\Ab)=1$. Also note that, for any choice of $\Cc$, $\x_0$ is orthogonal to \eqref{bad cone} by definition.

\samet{

\subsection{Proof of Theorem \ref{main31}}

Rotational invariance of Gaussian measurements allow us to make full use of Proposition \ref{bound2}. The following is a generalization of Theorem \ref{main31}.

\begin{proposition}\label{g main31} Consider the setup in Proposition \ref{bound2} where $\Meb$ has i.i.d $\Nc(0,1)$ entries. Let,
\beq
m_{low}=\frac{n(1-\bDb(\Cc))\ang(\Rc,\pa f(\x_0))^2}{100},\nn
\eeq
and suppose $\dim(\Rc)\leq m_{low}$. Then, whenever $m\leq m_{low}$, with probability $1-10\exp(-\frac{1}{16}\min\{m_{low},(1-\bDb(\Cc))^2n\})$, \eqref{rec_class} will fail for all functions $f(\cdot)$. 

\end{proposition}
\begin{proof}
More measurements can only increase the chance of success. Hence, without losing generality, assume $m=m_{low}$ and $\text{dim}(\Rc)\leq m$. The result will follow from Proposition \ref{bound2}. Recall that $m\leq \frac{(1-\bDb(\Cc)) n}{100}$.

\begin{itemize}
\item $\Pc_\Rc(\Meb^T)$ is statistically identical to a $\text{dim}(\Rc)\times m$ matrix with i.i.d. $\Nc(0,1)$ entries under proper unitary rotation. Hence, using Corollary 5.35 of \cite{Vershynin}, with probability $1-2\exp(-\frac{m}{8})$, $\sigma_{\max}(\Pc_\Rc(\Meb^T))\leq 1.5\sqrt{m}+\sqrt{\text{dim}(\Rc)}\leq 2.5\sqrt{m}$. With the same probability, $\sigma_{\min}(\Meb^T)\geq \sqrt{n}-1.5\sqrt{m}$.

\item From Theorem \ref{proj_norm}, using $m\leq \frac{(1-\bDb(\Cc))n}{100}$, with probability $1-6\exp(-\frac{(1-\bDb(\Cc))^2n}{16})$, $\rsing^2(\Meb^T)\geq \frac{1-\bDb(\Cc)}{4(1+\bDb(\Cc))}\geq \frac{1-\bDb(\Cc)}{8}$.

\end{itemize}
 Since $\frac{m}{n}\leq \frac{1}{30}$, combining these, with the desired probability,
\beq
\frac{ \sigma_{\max}(\Pc_\Rc(\Meb^T))}{\rsing(\Meb^T)\sigma_{\min}(\Meb^T)}\leq\sqrt{\frac{8}{1-\bDb(\Cc)}} \frac{2.5\sqrt{m}}{\sqrt{n}-1.5\sqrt{m}}< \frac{10\sqrt{m}}{\sqrt{(1-\bDb(\Cc))n}}.\nn
\eeq
Finally, using Proposition \ref{bound2} and $m\leq \frac{n(1-\bDb(\Cc))}{100} \ang(\Rc,\pa f(\x_0))^2$, with the same probability \eqref{rec_class} fails.
\end{proof}


To achieve Theorem \ref{main31}, choose $\Rc=\text{span}(\{\x_0\})$ and use the first statement of Proposition \ref{low subgrad}.

To achieve Corollary \ref{maincor3}, choose $\Rc=\text{span}(\{\x_0\})$ and use the second statement of Proposition \ref{low subgrad}.

}

\subsection{Enhanced lower bounds}\label{decomposable}

From our initial results, it may look like our lower bounds are suboptimal. For instance, considering only $\ell_1$ norm, $\kappa=\frac{\|\bx_0\|_1}{\sqrt{n}}$ lies between $\frac{1}{\sqrt{n}}$ and $\sqrt{\frac{k}{n}}$ for a $k$ sparse signal. Combined with Theorem \ref{main31}, this would give a lower bound of $\|\bx_0\|_1^2$ measurements. On the other hand, clearly, we need at least ${\cal{O}}(k)$ measurements to estimate a $k$ sparse vector. 

Indeed, Proposition \ref{bound2} gives such a bound with a better choice of $\Rc$. In particular, let us choose $\Rc=\text{span}(\{\text{sign}(\x_0)\})$. For any $\g\in\pa \|\x_0\|_1$, we have that,
\beq
\frac{\li\g,\frac{\text{sign}(\x_0)}{\sqrt{k}}\ri}{L}=\sqrt{\frac{k}{n}}\implies \ang(\text{sign}(\x_0),\pa \|\x_0\|_1)=\sqrt{\frac{k}{n}}\nn
\eeq 
Hence, we immediately have $m\geq \order{k}$ as a lower bound. The idea of choosing such sign vectors can be generalized to the so-called decomposable norms.

\begin{definition}[Decomposable Norm] \label{defdec}
A norm $\norm{\cdot}$ is decomposable at $\x\in\R^n$ if there exist a
subspace $T\subset \R^n$ and a vector $\e\in T$ such that the subdifferential at $\x$ has the form
	\begin{equation}
	\pa \norm{\x}=\{\z\in\R^n \;:\; \Pc_T(\z)=\e \; , \; \norm{\mathcal{P}_{T^\perp}(\z)}^*\leq 1\}.
\label{maindec}\nn
\end{equation}
We refer to $T$ as the \emph{support} and $\e$ as the \emph{sign vector} of $\x$ with respect to $\nor \, $.
\end{definition}

Similar definitions are used in \cite{candes_recht12} and \cite{wright12}. Our definition is simpler and less strict compared to these works. Note that $L$ is a global property of the norm while $\e$ and $T$ depend on both the norm and the point under consideration (decomposability is a local property in this sense).


To give some intuition for Definition \ref{defdec}, we review examples of norms that arise when considering simultaneously sparse and low rank matrices. For a matrix $\X\in \R^{\md_1\times \md_2}$, let $\X_{i,j}$, $\X_{i,.}$ and $\X_{.,j}$ denote its $(i,j)$ entry, $i$th row and $j$th column respectively.
\begin{lemma}[see \cite{candes_recht12}] \label{lemdecomp}The $\ell_1$ norm, the $\ell_{1,2}$ norm and the nuclear norm are decomposable as follows.\vspace{2pt}\\
{\bf{$\bullet$ {$\mathbf{\ell_1}$ norm}}} is decomposable at every $\x\in\R^n$, with sign $\e=\sgn{\x}\,$, and support as
    \[
    T=\supp{\x}=\left\{ \y\in \R^n :\;\; \x_i=0 \;\Rightarrow\; \y_i =0 \;\; \text{for } i=1,\ldots,n   \right\} \,.
    \]

\noindent {\bf{$\bullet$ {$\mathbf{\ell_{1,2}}$ norm}}} is decomposable at every $\X\in\R^{\md_1\times \md_2}$. The support is
    \[
    T= \left\{ \Y\in\R^{\md_1\times \md_2} : \;\; \X_{.,i}=\mathbf{0} \;\Rightarrow\;  \Y_{.,i}=\mathbf{0} \;\; \text{for } i=1,\ldots,\md_2   \right\},
    \]
     and the sign vector $\e\in \R^{\md_1\times \md_2}$ is obtained by normalizing the columns of $\X$ present in the support, $\e_{.,j} = \frac{\X_{.,j}}{\twonorm{\X_{.,j}}} \quad \text{if}\;\; \twonorm{\X_{.,j}}\neq 0$,
    and setting the rest of the columns to zero.

\vspace{5pt}
\noindent{\bf{$\bullet$ {Nuclear norm}}} is decomposable at every $\X\in\R^{\md_1\times \md_2}$. For a matrix $\X$ with rank $r$ and compact singular value decomposition $\X=\U\bSi\V^T$ where $\bSi\in\R^{r\times r}$, we have $\e = \U\V^T$ and
    \beqas
    T &= \left\{\Y\in\R^{\md_1\times \md_2} : \;\; (\Ib-\U\U^T)\Y(\Ib-\V\V^T)=\mathbf{0}  \right\}  \\
    &=\left\{ \Z_1\V^T + \U\Z_2^T \mid \Z_1 \in \R^{\md_1\times r}, \Z_2 \in \R^{\md_2\times r} \right\}.
    \eeqas
    \end{lemma}

The next lemma shows that the sign vector $\e$ will yield the largest $\corr$ with the subdifferential and the best lower bound for such norms.
\begin{lemma} \label{e is best}Let $\|\cdot\|$ be a decomposable norm with support $T$ and sign vector $\e$. For any $\vb\neq 0$, we have that,
\beq
\ang(\vb,\pa\|\x_0\|)\leq \ang({\e},\pa\|\x_0\|)\label{sign bound}
\eeq
Also $\ang({\e},\pa\|\x_0\|)\geq \frac{\|\e\|_2}{L}$.
\end{lemma}
\begin{proof} Let $\vb$ be a unit vector. Without losing generality, assume $\vb^T\e\geq 0$. Pick a vector $\z\in T^{\perp}$ with $\|\z\|^*=1$ such that $\z^T\vb\leq 0$ (otherwise pick $-\z$). Now, consider the class of subgradients $\g(\alpha)=\e+\alpha \z$ for $1\geq \alpha\geq -1$. Then,
\beq
\inf_{-1\leq \alpha\leq 1}\frac{|\vb^T\g(\alpha)|}{\|\g(\alpha)\|_2}=\inf_{0\leq \alpha\leq 1}\frac{|\vb^T\g(\alpha)|}{\|\g(\alpha)\|_2}=\inf_{0\leq \alpha\leq 1}\frac{|\e^T\vb-\alpha |\z^T\vb||}{(\|\e\|_2^2+\alpha^2\|\z\|_2^2)^{1/2}}\nn
\eeq
If $|\z^T\vb|\geq \e^T\vb$, then, the numerator can be made $0$ and $\ang(\vb,\pa \|\x_0\|)=0$. Otherwise, the right hand side is decreasing function of $\alpha$, hence the minimum is achieved at $\alpha=1$, which gives,
\beq
\inf_{-1\leq \alpha\leq 1}\frac{|\vb^T\g(\alpha)|}{\|\g(\alpha)\|_2}=\frac{|\e^T\vb- |\z^T\vb||}{(\|\e\|_2^2+\|\z\|_2^2)^{1/2}}\leq \frac{|\e^T\vb|}{(\|\e\|_2^2+\|\z\|_2^2)^{1/2}}\leq \frac{\|\e\|_2}{(\|\e\|_2^2+\|\z\|_2^2)^{1/2}}=\inf_{-1\leq \alpha\leq 1}\frac{|{\bar{\e}}^T\g(\alpha)|}{\|\g(\alpha)\|_2}\nn
\eeq
where we used $\e^T\g(\alpha)=\e^T\e=\|\e\|^2_2$. Hence, along any direction $\z$, $\e$ yields a higher minimum correlation than $\vb$. To obtain \eqref{sign bound}, further take infimum over all $\z\in T^{\perp},\|\z\|^*\leq 1$ which will yield infimum over $\pa \|\x_0\|$. Finally, use $\|g(\alpha)\|_2\leq L$ to lower bound $\ang({\e},\pa\|\x_0\|)$.

%
\end{proof}

%

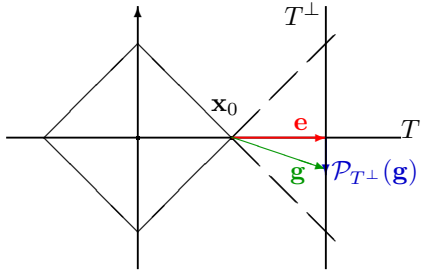
\begin{figure}[t!]
\begin{center}
\setlength{\unitlength}{0.25cm}
\vspace{-.2in}
\begin{picture}(35,-2)
	\color{black}
	\linethickness{.25mm}
	\put(16,-17){\vector(0,1){14}}	
	\put(9,-10){\line(1,0){21}}		
\put(30,-10){$T$}
	\linethickness{.5mm}
	\put(15.9,-10){\line(1,0){.2}}		
	\put(20.9,-10){\line(1,0){.2}}		
	\put(19.9,-8.6){$\x_0$}			

	\linethickness{.25mm}
	\put(11,-10){\line(1,1){5}}		
	\put(21,-10){\line(-1,1){5}}		
	\put(11,-10){\line(1,-1){5}}		
	\put(21,-10){\line(-1,-1){5}}		

\multiput(21,-10)(2,2){3}{\line(1,1){1.5}}		
\multiput(21,-10)(2,-2){3}{\line(1,-1){1.5}}		

\color{red}
	\put(21,-10){\vector(1,0){5}}		
	\put(24.3,-9.6){$\e$}			
\color{darkgreen}
	\put(21,-10){\vector(3,-1){5}}		
	\put(24.1,-12.2){$\g$}			
\color{darkblue}
	\put(26.0,-10){\vector(0,-1){2}}		
	\put(26.3,-12,2){$\Pc_{T^\perp}(\g)$}			

\color{black}
\put(26,-10){\line(0,1){7}}
\put(26,-10){\line(0,-1){7}}
\put(23.8,-4){$T^\perp$}
\end{picture}
\end{center}
\vspace{1.65in}
\caption{An example of a decomposable norm: $\ell_1$ norm is decomposable at $\x_0=(1,0)$. The sign vector $\e$, the support $T$, and shifted subspace $T^\perp$ are illustrated. A subgradient $\g$ at $\x_0$ and its projection onto $T^\perp$ are also shown.
}
\end{figure}

Based on Lemma \ref{e is best}, the individual lower bound would be $\order{\frac{\|\e\|_2^2}{L^2}} n$. Calculating $\frac{\|\e\|_2^2}{L^2}n$ for the norms in Lemma \ref{lemdecomp}, reveals that, this quantity is $k$ for a $k$ sparse vector, $c\md_1$ for a $c$-column sparse matrix and $r\max\{\md_1,\md_2\}$ for a rank $r$ matrix. Compared to bounds obtained by using $\bx_0$, these new quantities are directly proportional to the true model complexities. Finally, we remark that, these new bounds correspond to choosing $\x_0$ that maximizes the value of $\|\bx_0\|_{1},\|\bx_0\|_{\star}$ or $\|\bx_0\|_{1,2}$ while keeping sparsity, rank or column sparsity fixed. In particular, in these examples, $\e$ has the same sparsity, rank, column sparsity as $\x_0$.

The next lemma gives a $\corr$ bound for the combination of decomposable norms as well as a simple lower bound on the sample complexity.
\begin{proposition}\label{decomp prop} Given decomposable norms $\|\cdot\|_{(i)}$ with supports $T_i$ and sign vectors $\e_i$. Let $T_{\cap}= \bigcap_{1\leq i\leq \tau} T_i$. Choose the subspace $\Rc$ to be a subset of $T_\cap$. 
\begin{itemize}
\item Assume $\li\Pc_{\Rc}(\e_i),\Pc_{\Rc}(\e_j)\ri\geq 0$ for all $i,j$ and $\min_{1\leq i\leq \tau} \frac{\|\Pc_{\Rc}(\e_i)\|_2}{\|\e_i\|_2}\geq \upsilon$. Then,
\beq
\ang(\Rc,\pa f(\x_0))\geq \frac{\upsilon}{\sqrt{\tau}}\min_{1\leq i\leq \tau} \ang(\e_i,\pa \|\x_0\|_{(i)}).\nn
\eeq
\item Consider Proposition \ref{bound2} with Gaussian measurements and suppose $\Rc$ is orthogonal to the set \eqref{bad cone}. Let $f(\x)=\sum_{i=1}^\tau \la_i\|\x\|_{(i)}$ for nonnegative $\{\la_i\}$'s. Then, if $m< \text{dim}(\Rc)$, \eqref{rec_class} fails with probability $1$.
\end{itemize}
\end{proposition}
\begin{proof} Let $\g=\sum_{i=1}^\tau w_i\g_i$ for some $\g_i\in \pa \|\x_0\|_{(i)}$. First, $\|\g\|_2\leq \sum_{i=1}^\tau w_i\|\g_i\|_2$. Next,
\beq
\|\Pc_{\Rc}(\g)\|_2^2=\|\sum_{i=1}^\tau w_i\Pc_{\Rc}(\e_i)\|_2^2\geq \sum_{i=1}^\tau w_i^2\|\Pc_{\Rc}(\e_i)\|_2^2\geq \upsilon^2\sum_{i=1}^\tau w_i^2\|\e_i\|_2^2\geq \frac{\upsilon^2}{\tau}(\sum_{i=1}^\tau w_i\|\e_i\|_2)^2.\nn
\eeq
To see the second statement, consider the line \eqref{look here} from the proof of Proposition \ref{bound2}. $\Pc_\Rc(\g)=\sum_{i=1}^\tau \la_i\Pc_\Rc(\e_i)$. On the other hand, column space of $\Pc_\Rc(\Meb^T)$ is an $m$-dimensional random subspace of $\Rc$. If $m< \text{dim}(\Rc)$, $\Pc_\Rc(\g)$ is linearly independent with $\Pc_\Rc(\Meb^T)$ with probability $1$ and \eqref{look here} will not hold.
\end{proof}


In the next section, we will show how better choices of $\Rc$ (based on the decomposability assumption) can improve the lower bounds for S\&L recovery.

\section{Proofs for Section \ref{sub:SLR}} \label{sec:SLR}
Using the general framework provided in Section \ref{generalres}, in this section we present the proof of Theorem \ref{main_thm_sec3}, which states various convex and nonconvex recovery results for the S\&L models. We start with the proofs of the convex recovery.

\subsection{Convex recovery results for $S\&L$}
In this section, we prove the statements of Theorem \ref{main_thm_sec3} regarding convex approaches, using Theorem \ref{main31} and Proposition \ref{g main31}. We will make use of the decomposable norms to obtain better lower bounds. Hence, we first state a result on the sign vectors and the supports of the S\&L model following Lemma \ref{lemdecomp}. The proof is provided in Appendix \ref{app:norms}.
\begin{lemma} \label{signvec}
Denote the norm $\|\X^T\|_{1,2}$ by $\|\cdot^T\|_{1,2}$. Given a matrix $\X_0\in\R^{\md_1\times \md_2}$, let $\Eb_\st,\Eb_c,\Eb_r$ and $T_\st,T_c,T_r$ be the sign vectors and supports for the norms $\nor_\star$, $\nor_{1,2}$, $\|\cdot^T\|_{1,2}$ respectively. Then,
\begin{itemize}
\item $\Eb_\st,\Eb_r,\Eb_c \in  T_\st\cap T_c\cap T_r$,
\item $\iprod{\Eb_\st}{\Eb_r}\geq 0$, $\iprod{\Eb_\st}{\Eb_c}\geq 0$, and $\iprod{\Eb_c}{\Eb_r}\geq 0$.
\end{itemize}
\end{lemma}

\subsubsection{Proof of Theorem \ref{main_thm_sec3}: Convex cases}

\paragraph{Proof of (a1)} We use the functions $\|\cdot\|_{1,2},\|\cdot^T\|_{1,2}$ and $\|\cdot\|_\st$ without the cone constraint, i.e., $\Cc=\R^{\md_1\times \md_2}$. We will apply Proposition \ref{g main31} with $\Rc=T_\st\cap T_c\cap T_r$. From Lemma \ref{signvec} all the sign vectors lie on $\Rc$ and they have pairwise nonnegative inner products. Consequently, applying Proposition \ref{decomp prop},
\beq
\ang(\Rc,\pa f(\X_0))^2\geq \frac{1}{3} \min\{\frac{k_1}{\md_1},\frac{k_2}{\md_2},\frac{r}{\min\{\md_1,\md_2\}}\}\nn
\eeq
If $m<\text{dim}(\Rc)$, we have failure with probability $1$. Hence, assume $m\geq \text{dim}(\Rc)$. Now, apply Proposition \ref{g main31} with the given $m_{low}$. 

\paragraph{Proof of (b1)} In this case, we apply Lemma \ref{useful psd}. We choose $\tcs=T_\st\cap T_c\cap T_r \cap \Sbb^n$, the norms are the same as in the general model, and $\upsilon\geq\frac{1}{\sqrt{2}}$. Also, pairwise  inner products are positive, hence, using Proposition \ref{decomp prop}, $\ang(\Rc,\pa f(\X_0))^2\geq \frac{1}{4} \min\{\frac{k}{\md},\frac{r}{\md}\}$. Again, we may assume $m\geq \text{dim}(\Rc)$. Finally, based on Corollary \ref{psdcorol}, for the PSD cone we have $\bDb(\Cc)\geq\frac{\sqrt{3}}{2}$. The result follows from Proposition \ref{g main31} with the given $m_{low}$.

\paragraph{Proof of (c1)}
For PSD cone, $\bDb(\Cc)\geq\frac{\sqrt{3}}{2}$ and we simply use Theorem \ref{main31} to obtain the result by using $\kappa_{\ell_1}^2=\frac{\|\bX_0\|_1^2}{\md^2}$ and $\kappa_{\star}^2=\frac{\|\bX_0\|_\star^2}{\md}$.

\subsubsection{Proof of Corollary \ref{weighted s&l}}
To show this, we will simply use Theorem \ref{bound} and will substitute $\kappa$'s corresponding to $\ell_1$ and the nuclear norm. $\kappa_\st=\frac{\|{\bar{\X}}_0\|_\star}{\sqrt{\md}}$ and $\kappa_{\ell_1}=\frac{\|{\bar{\X}}_0\|_{\ell_1}}{\md}$. Also observe that, $\la_{\ell_1} L_{\ell_1}=\beta \md$ and $\la_{\star}L_{\star}=(1-\beta) \md$. Hence, $\sum_{i=1}^2 {\bar{\la}}_i \kappa_i=\alpha \|{\bar{\X}}_0\|_1+(1-\alpha)\|{\bar{\X}}_0\|_\star\sqrt{\md}$. Use Proposition \ref{sub gauss} to conclude with sufficiently small $c_1,c_2>0$.

\subsection{Nonconvex recovery results for S\&L}
While Theorem \ref{main_thm_sec3} states the result for Gaussian measurements, we prove the nonconvex recovery for the more general sub-gaussian measurements. We first state a lemma that will be useful in proving the nonconvex results. The proof is provided in the Appendix \ref{appC} and uses standard arguments.
\begin{lemma}\label{abcgcc} Consider the set of matrices $M$ in $\R^{\md_1\times \md_2}$ that are supported over an $s_1\times s_2$ submatrix with rank at most $q$. There exists a constant $c>0$ such that whenever $m\geq  c\min\{(s_1+s_2)q,s_1\log \frac{\md_1}{s_1},s_2\log \frac{\md_2}{s_2}\}$, with probability $1-2\exp(-cm)$, $\Mec(\cdot):\R^{\md_1\times \md_2}\rightarrow \R^m$ with i.i.d. zero-mean and isotropic sub-gaussian rows will satisfy the following,
\beq\label{abcgcceq}
\Mec(\X)\neq 0,~~~\text{for all}~~~\X\in M.
\eeq
\end{lemma}

\subsubsection{Proof of Theorem \ref{main_thm_sec3}: Nonconvex cases}
Denote the sphere in $\R^{\md_1\times \md_2}$ with unit Frobenius norm by $\Sc^{\md_1\times \md_2}$.

\paragraph{Proof of (a2)} Observe that the function $f(\X)=\frac{\|\X\|_{0,2}}{\|\X_0\|_{0,2}}+\frac{\|\X^T\|_{0,2}}{\|\X_0^T\|_{0,2}}+\frac{\text{rank}(\X)}{\text{rank}(\X_0)}$ satisfies the triangle inequality and we have $f(\X_0)=3$. Hence, if all null space elements $\W\in\Null(\Mec)$ satisfy $f(\W)>6$, we have
\beqas
f(\X)\geq f(\X-\X_0)-f(-\X_0)>3,
\eeqas
for all feasible $\X$ which implies $\X_0$ being the unique minimizer.

Consider the set $M$ of matrices, which are supported over a $6k_1\times 6k_2$ submatrix with rank at most $6r$. Observe that any $\Z$ satisfying $f(\Z)\leq 6$ belongs to $M$. Hence ensuring $\Null(\Mec)\cap M = \{0\}$ would ensure $f(\W)>6$ for all $\W\in\Null(\Mec)$. Since $M$ is a cone, this is equivalent to $\Null(\Mec)\cap (M \cap \Sc^{\md_1\times \md_2}) = \emptyset$.
Now, applying Lemma \ref{abcgcc} with set $M$ and $s_1=6k_1$, $s_2=6k_2$, $q=6r$ we find the desired result.
%

\paragraph{Proof of (b2)} Observe that due to the symmetry constraint,
\[
f(\X)=\frac{\|\X\|_{0,2}}{\|\X_0\|_{0,2}}+\frac{\|\X^T\|_{0,2}}{\|\X_0^T\|_{0,2}}+
\frac{\text{rank}(\X)}{\text{rank}(\X_0)}.
\]
Hence, the minimization is the same as {\bf{(a2)}}, the matrix is rank $r$ contained in a $k\times k$ submatrix and we additionally have the positive semidefinite constraint which can only reduce the amount of required measurements compared to {\bf{(a2)}}. Consequently, the result follows by applying Lemma \ref{abcgcc}, similar to {\bf{(a2)}}.


\paragraph{Proof of (c2)} Let $C=\{\X\neq 0 \big| f(\X)\leq f(\X_0)\}$. Since $\text{rank}(\X_0)=1$, if $f(\X)\leq f(\X_0)=2$, $\text{rank}(\X)=1$. With the symmetry constraint, this means $\X=\pm\x\x^T$ for some $l$-sparse $\x$. Observe that $\X-\X_0$ has rank at most $2$ and is contained in a $2k\times 2k$ submatrix as $l\leq k$. Let $M$ be the set of matrices that are symmetric and whose support lies in a $2k\times 2k$ submatrix. Using Lemma \ref{abcgcc} with $q=2$, $s_1=s_2=2k$, whenever $m\geq ck\log \frac{n}{k}$, with desired probability all nonzero $\W\in M$ will satisfy $\Ac(\W)\neq 0$. Consequently, any $\X\in C$ will have $\Ac(\X)\neq \Ac(\X_0)$, hence $\X_0$ will be the unique minimizer.

\subsection{Existence of a matrix with large $\kappa$'s}\label{explicit construction}

We now argue that, there exists an $S\&L$ matrix that have large $\kappa_{\ell_1},\kappa_{\ell_{1,2}}$ and $\kappa_{\star}$ simultaneously. We will have a deterministic construction that is close to optimal. Our construction will be based on Hadamard matrices. $\Hb_n\in\R^{n\times n}$ is called a Hadamard matrix if it has $\pm 1$ entries and orthogonal rows. Hadamard matrices exist for $n$ that is an integer power of $2$. 


Using $\Hb_n$, our aim will be to construct a $\md_1\times \md_2$ $S\&L$ $(k_1,k_2,r)$ matrix $\X_0$ that satisfy $\|\bX_0\|_1^2\approx k_1k_2$, $\|\bX_0\|_\star^2 \approx r$, $\|\bX_0\|_{1,2}^2 \approx k_2$ and $\|\bX_0^T\|_{1,2}^2 \approx k_1$. To do this, we will construct a $k_1\times k_2$ matrix and then plant it into a larger $\md_1\times \md_2$ matrix. The following lemma summarizes the construction.

\begin{lemma} Without loss of generality, assume $k_2\geq k_1\geq r$. Let $\Hb:=\Hb_{\lfloor\log_2 k_2\rfloor}$. Let $\X\in\R^{k_1\times k_2}$ be so that, $i$'th row of $\X$ is equal to $[i-1~(\text{mod}~r)]+1$'th row of $\Hb$ followed by $0$'s for $1\leq i\leq k_1$. Then,
\beq
\|\bX_0\|_1^2\geq  \frac{k_1k_2}{2},~\|\bX_0\|_\star^2 \geq \frac{r}{2},~\|\bX_0\|_{1,2}^2 \geq \frac{ k_2}{2},~\|\bX_0^T\|_{1,2}^2 = k_1.\nn
\eeq
In particular, if $k_1\equiv 0~(\text{mod}~r)$ and $k_2$ is an integer power of $2$, then,
\beq
\|\bX_0\|_1^2= k_1k_2,~\|\bX_0\|_\star^2 = r,~\|\bX_0\|_{1,2}^2 = k_2,~\|\bX_0^T\|_{1,2}^2 = k_1.\nn
\eeq
\end{lemma}
\begin{proof}
The left $k_1\times 2^{\lfloor\log_2 k_2\rfloor}$ entries of $\X$ are $\pm1$, and the remaining entries are $0$. This makes the calculation of $\ell_1$ and $\ell_{1,2}$ and Frobenius norms trivial. 

In particular, $\|\X_0\|_F^2=\|\X_0\|_1=k_1 2^{\lfloor\log_2 k_2\rfloor}$, $\|\X_0\|_{1,2}=\sqrt{k_1}2^{\lfloor\log_2 k_2\rfloor}$ and $\|\X_0^T\|_{1,2}=k_12^{\frac{\lfloor\log_2 k_2\rfloor}{2}}$. Substituting these yield the results for these norms.

To lower bound the nuclear norm, observe that, each of the first $r$ rows of the $\Hb$ are repeated at least $\lfloor\frac{k_1}{r}\rfloor$ times in $\X$. Combined with the orthogonality, this ensures that each singular value of $\X$ that is associated with the $j$'th row of $\Hb$ is at least $\sqrt{2^{\lfloor\log_2 k_2\rfloor}\lfloor\frac{k_1}{r}\rfloor}$ for all $1\leq j\leq r$. Consequently,
\beq
\|\X\|_\star\geq r\sqrt{2^{\lfloor\log_2 k_2\rfloor}\lfloor\frac{k_1}{r}\rfloor}\nn
\eeq
Hence,
\beq
\|\bX\|_\star \geq \frac{r\sqrt{2^{\lfloor\log_2 k_2\rfloor}\lfloor\frac{k_1}{r}\rfloor}}{\sqrt{k_1 2^{\lfloor\log_2 k_2\rfloor}}}=\frac{r\sqrt{2^{\lfloor\log_2 k_2\rfloor}\lfloor\frac{k_1}{r}\rfloor}}{\sqrt{ 2^{\lfloor\log_2 k_2\rfloor}}}=r\sqrt{\frac{1}{k_1}\lfloor\frac{k_1}{r}\rfloor}\nn
\eeq
Use the fact that $\lfloor\frac{k_1}{r}\rfloor\geq \frac{k_1}{2r}$ as $k_1\geq r$.
\end{proof}
If we are allowed to use complex numbers, one can apply the same idea with the Discrete Fourier Transform (DFT) matrix. Similar to $\Hb_n$, DFT has orthogonal rows and its entries have the same absolute value. However, it exists for any $n\geq 1$; which would make the argument more concise.

%
%
%
%
%
%
%
%
%

\section{Numerical Experiments} \label{sec:numerical}
\newcommand{\fbestone}{\max\{\frac{\tr{\X}}{\tr{\X_0}} , \frac{\norm{\X}_1}{\norm{\X_0}_1}\}}
\newcommand{\fbestonetwo}{\max\{\frac{\tr{\X}}{\tr{\X_0}} , \frac{\norm{\X}_{1,2}}{\norm{\X_0}_{1,2}}\}}

In this section, we numerically verify our theoretical bounds on the number of measurements for the Sparse and Low-rank recovery problem. We demonstrate the empirical performance of the weighted maximum of the norms $f\best$ (see Lemma \ref{best opt}), as well as the weighted sum of norms.

The experimental setup is as follows. Our goal is to explore how the number of required measurements $m$ scales with the size of the matrix $\md$. We consider a grid of $(m,\md)$ values, and generate at least 100 test instances for each grid point
(in the boundary areas, we increase the number of instances to at least 200).

\begin{figure}[t!]
\centering
\psfrag{$\md$}{$\md$}
\psfrag{m}{$m$}
\includegraphics[scale=0.5]{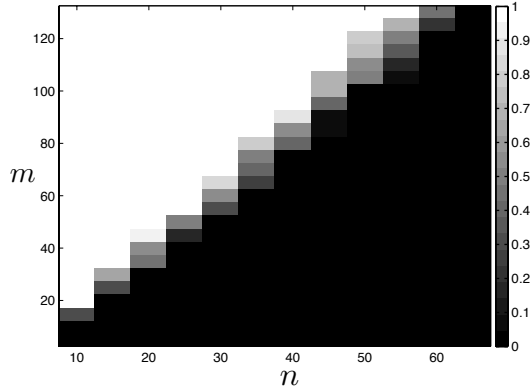}
\caption{Performance of the recovery program minimizing $\fbestonetwo$ with a PSD constraint. The dark region corresponds to the experimental region of failure due to insufficient measurements. As predicted by Theorem \ref{main_thm_sec3}, the number of required measurements increases linearly with $r\md$.}
\label{fig:L12_best_psd}
\end{figure}
We generate the target matrix $\X_0$ by generating a $k\times r$ i.i.d.~Gaussian matrix $\mathbf{G}$, and inserting the $k\times k$ matrix $\mathbf{G}\mathbf{G}^T$ in an $\md\times \md$ matrix of zeros.
We take $r=1$ and $k=8$ in all of the following experiments; even with these small values, we can observe the scaling predicted by our bounds.
In each test, we measure the normalized recovery error $\frac{\fronorm{\X-\X_0}}{\fronorm{\X_0}}$ and declare successful recovery when this error is less than $10^{-4}$.
The optimization programs are solved using the CVX package \cite{cvx}, which calls the SDP solver SeDuMi \cite{sedumi}.

We first test our bound in part (b) of Theorem \ref{main_thm_sec3}, $\Omega(r \md)$, on the number of measurements for recovery in the case of minimizing $\fbestonetwo$ over the set of positive semi-definite matrices. Figure \ref{fig:L12_best_psd} shows the results, which demonstrates $m$ scaling linearly with $\md$ (note that $r=1$).

Next, we replace $\ell_{1,2}$ norm with $\ell_1$ norm and consider a recovery program that emphasizes entry-wise sparsity rather than block sparsity.
Figure \ref{fig:L1_best_psd} demonstrates the lower bound $\Omega(\min\{k^2, \md\})$ in Part (c) of Theorem \ref{main_thm_sec3} where we attempt to recover a rank-1 positive semi-definite matrix $\X_0$ by minimizing $\fbestone$ subject to the measurements and a PSD constraint.
%
%
The green curve in the figure shows the empirical 95\% failure boundary, depicting the region of failure with high probability that our results have predicted. It starts off growing linearly with $\md$, when the term $r\md$ dominates the term $k^2$, and then saturates as $\md$ grows and the $k^2$ term (which is a constant in our experiments) becomes dominant.

\begin{figure}[t!]
\centering
\psfrag{$\md$}{$\md$}
\psfrag{m}{$m$}
\includegraphics[scale=0.5]{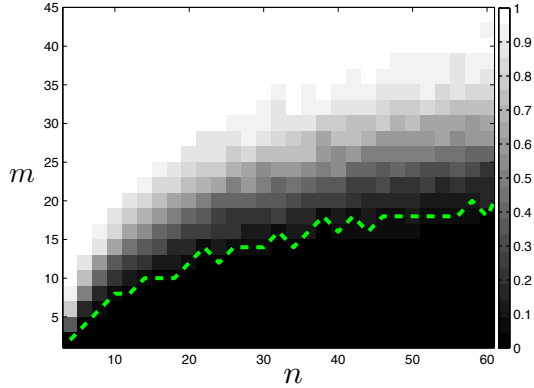}
\caption{Performance of the recovery program minimizing $\fbestone$ with a PSD constraint. $r=1,k=8$ and $\md$ is allowed to vary. The plot shows ${m}$ versus $\md$ to illustrate the lower bound $\Omega(\min\{k^2,\md r\})$ predicted by Theorem \ref{main_thm_sec3}. }
\label{fig:L1_best_psd}
\end{figure}

\begin{figure}[t!]
\begin{minipage}[b]{0.5\linewidth}
\centering
\psfrag{$\md$}{$\md$}
\psfrag{m}{${m}$}
\includegraphics[scale=0.5]{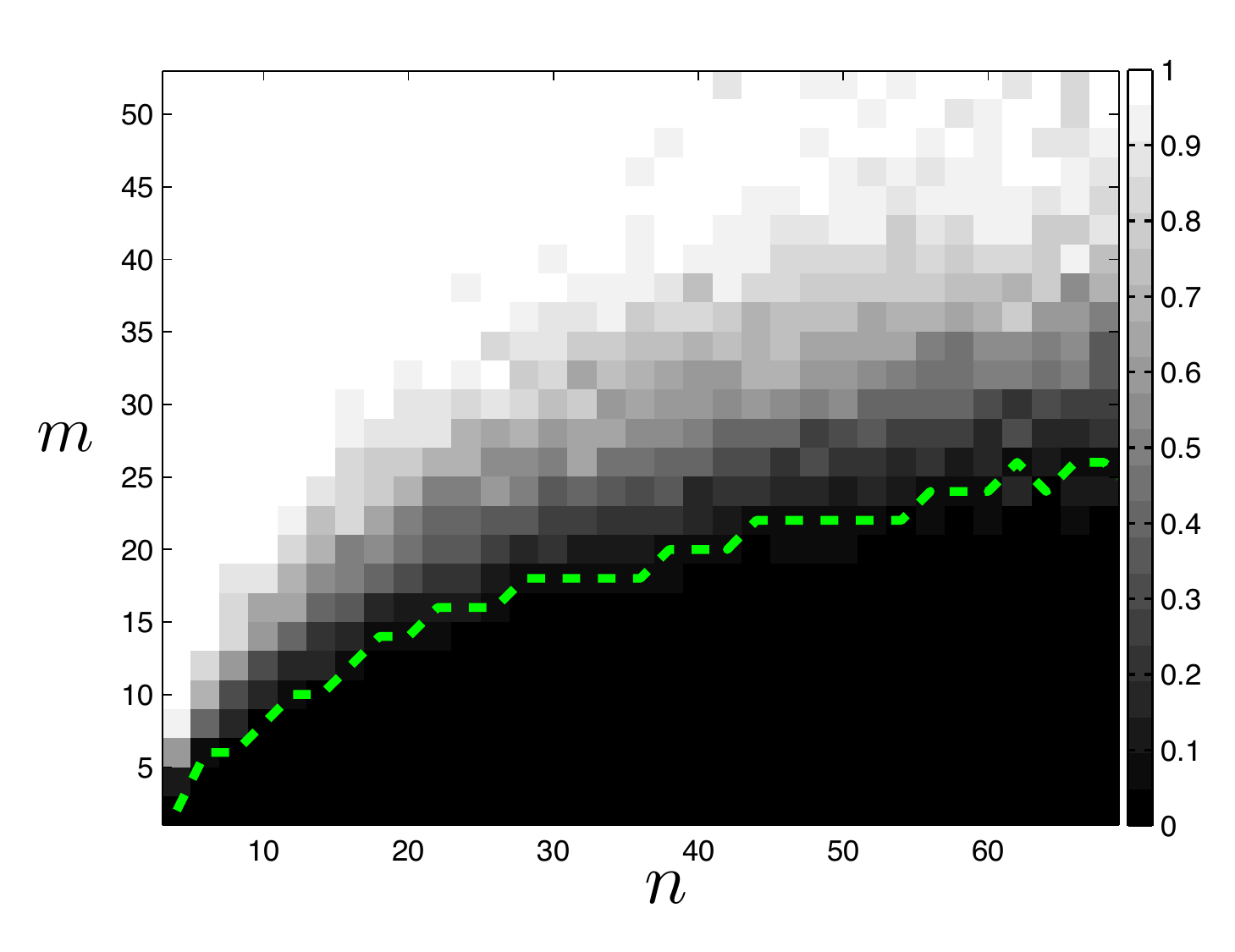}
\end{minipage}
\hspace{0.5cm}
\begin{minipage}[b]{0.5\linewidth}
\centering
\psfrag{$\md$}{$\md$}
\psfrag{m}{${m}$}
\includegraphics[scale=0.5]{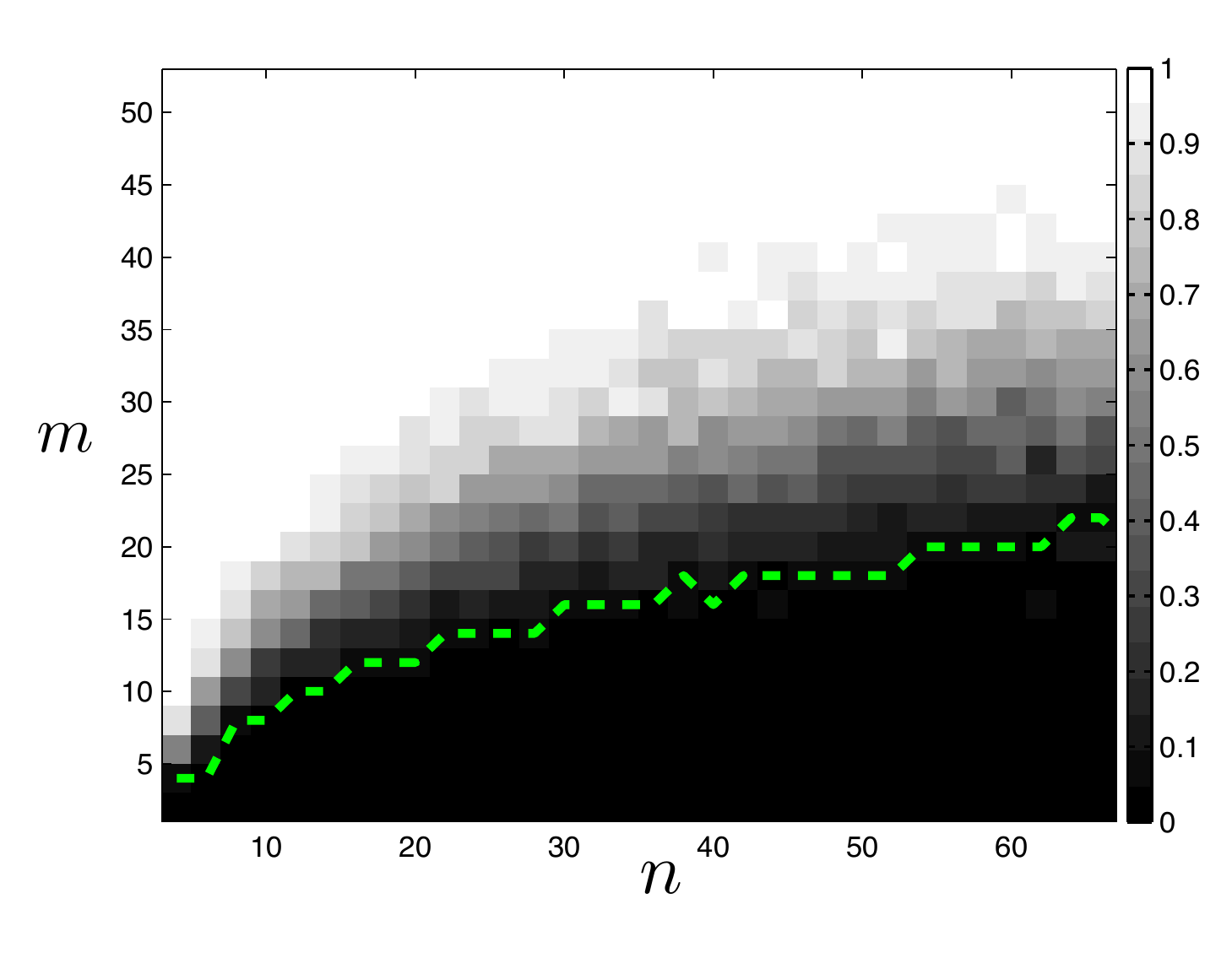}
\end{minipage}
\caption{Performance of the recovery program minimizing $\tr{\X}+\lambda \norm{\X}_1$ with a PSD constraint, for $\lambda= 0.2$ (left) and $\lambda = 0.35$ (right). }
\label{fig:penalty_compare}
\end{figure}


The penalty function $\fbestone$ depends on the norm of $\X_0$. In practice the norm of the solution is not known beforehand, a weighted sum of norms is used instead.
In Figure \ref{fig:penalty_compare} we examine the performance of the weighted sum of norms penalty in recovery of a rank-1 PSD matrix, for different weights. We pick $\lambda=0.20$ and $\lambda=0.35$ for a randomly generated matrix $\X_0$, and it can be seen that we get a reasonable result which is comparable to the performance of $\fbestone$.

In addition, we consider the \emph{amount of error} in the recovery when the program fails.
Figure \ref{twocurves} shows two curves below which we get a $90\%$ percent failure, where for the green curve the normalized error threshold for declaring failure is $10^{-4}$, and for the red curve
it is a larger value of $0.05$. We minimize $\fbestone$ as the objective.
We observe that when the recovery program has an error, it is very likely that this error is large, as the curves for $10^{-4}$ and $0.05$ almost overlap. Thus, when the program fails, it fails badly. This observation agrees with intuition from similar problems in compressed sensing where sharp phase transition is observed.
\begin{figure}[t!]
\centering
\psfrag{$\md$}{$\md$}
\psfrag{mt}{${m}$}
\includegraphics[scale=0.5]{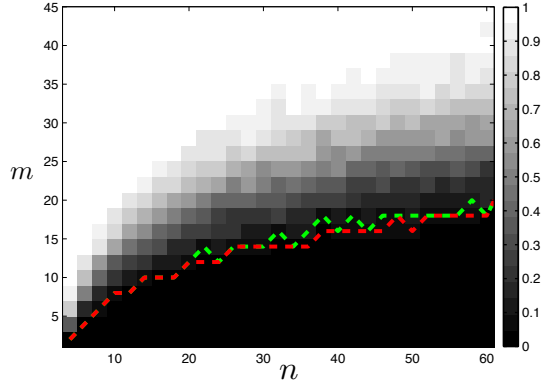}
\caption{90\% frequency of failure where the threshold of recovery is $10^{-4}$ for the green and $0.05$ for the red curve. $\fbestone$ is minimized subject to the PSD constraint and the measurements.}
\label{twocurves}
\end{figure}

As a final comment, observe that, in Figures \ref{fig:L1_best_psd}, \ref{fig:penalty_compare} and \ref{twocurves} the required amount of measurements slowly increases even when $\md$ is large and $k^2=64$ is the dominant constant term. While this is consistent with our lower bound of $\Omega(k^2,\md)$, the slow increase for constant $k$, can be explained by the fact that, as $\md$ gets larger, sparsity becomes the dominant structure and $\ell_1$ minimization by itself requires $\order{k^2\log\frac{\md}{k}}$ measurements rather than $\order{k^2}$. Hence for large $\md$, the number of measurements can be expected to grow logarithmically in $\md$.

\samet{
\begin{figure}[t!]
\centering

\includegraphics[scale=0.26]{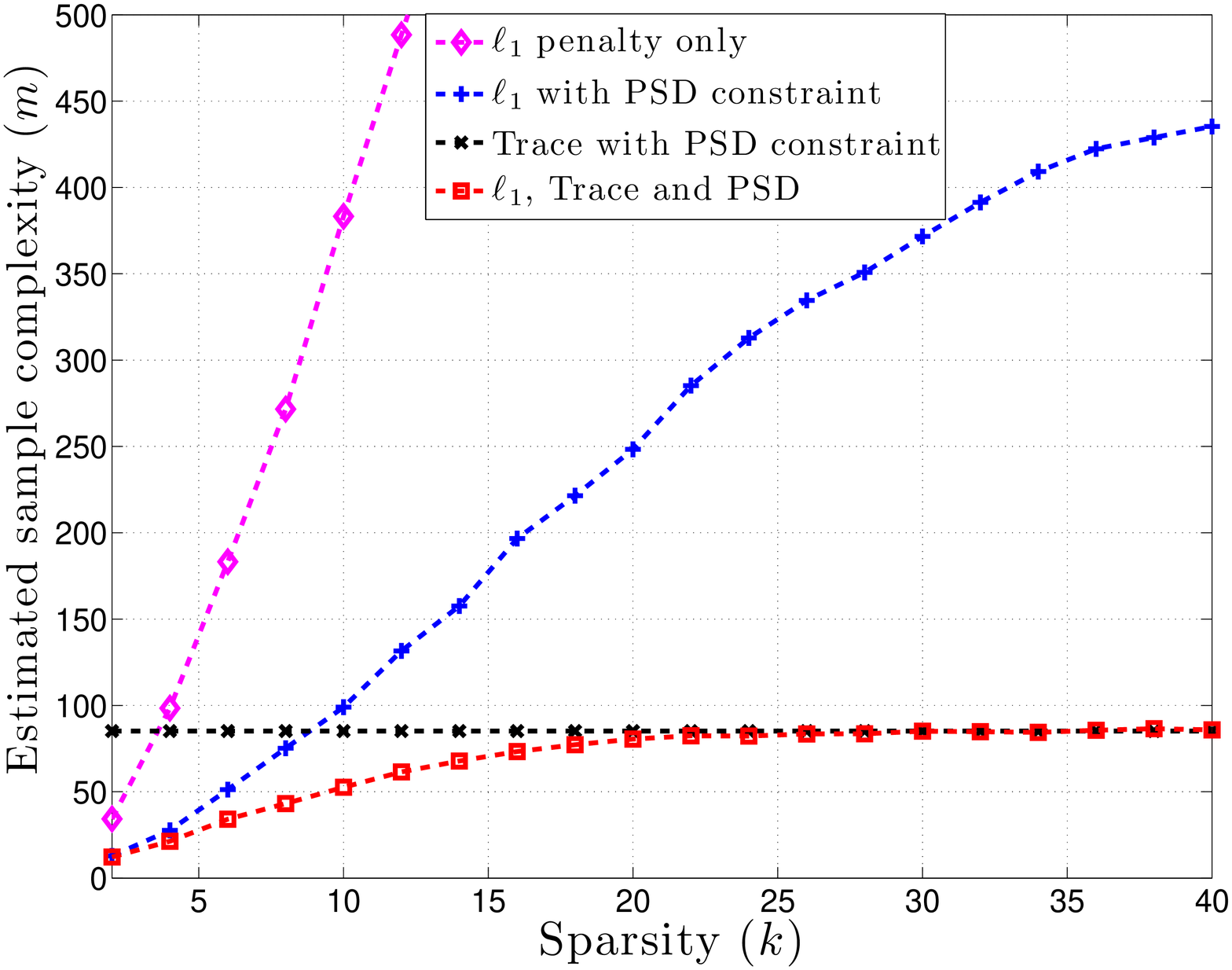}
\caption{We compare sample complexities of different approaches for a rank $1$, $40\times 40$ matrix as function of sparsity. The sample complexities were estimated by a search over $m$, where we chose the $m$ with success rate closest to $50\%$ (over $100$ iterations).}
\label{4compare}
\end{figure}

In Figure \ref{4compare}, we compare the estimated phase transition points for different approaches for varying sparsity levels. The algorithms we compare are,
\begin{itemize}
\item Minimize $\ell_1$ norm,
\item Minimize $\ell_1$ norm subject to the positive-semidefinite constraint,
\item Minimize trace norm subject to the positive-semidefinite constraint,
\item Minimize $\fbestone$ subject to the positive-semidefinite constraint
\end{itemize}

Not surprisingly, the last option outperforms the rest in all cases. On the other hand, its performance is highly comparable to the minimum of the second and third approaches. For all regimes of sparsity, we observe that, measurements required by the last method is at least half as much as the minimum of second and third methods.
}

\section{Discussion} \label{sec:disc}

We have considered the problem of recovery of a simultaneously structured object from limited measurements. It is common in practice to combine known norm penalties corresponding to the individual structures (also known as regularizers in statistics and machine learning applications), and minimize this combined objective in order to recover the object of interest. The common use of this approach motivated us to analyze its performance, in terms of the smallest number of generic measurements needed for correct recovery. We showed that, under a certain assumption on the norms involved, the combined penalty requires more generic measurements than one would expect based on the degrees of freedom of the desired object.
Our lower bounds on the required number of measurements implies that the combined norm penalty cannot perform significantly better than the best individual norm.

These results raise several interesting questions, and lead to directions for future work. We briefly outline some of these directions, as well as connections to some related problems.

\paragraph{Quantifying recovery failure via error bounds.}
We observe from the recovery error plots shown in Figure \ref{twocurves} that whenever our recovery program fails, it fails with a significant recovery error.
The figure shows two curves under which recovery fails with high probability, where failure is defined by the normalized error $\fronorm{\X-\X_0}/ \fronorm{\X_0}$ being above $10^{-4}$ and $0.05$.
The two curves almost coincide.
This observation leads to the question of whether we can characterize how large the error is with a high probability over the random measurements. A lower bound on the recovery error as a function of the number of problem parameters will be very insightful.

\paragraph{Defining new atoms for simultaneously structured models.}
Our results show that combinations of individual norms do not exhibit a strong recovery performance.
On the other hand, the seminal paper \cite{chandra} proposes a remarkably general construction for an appropriate penalty given a set of atoms. Can we revisit a simultaneously structured recovery problem, and define new atoms that capture all structures at the same time? And can we obtain a new norm penalty induced by the convex hull of the atoms? Abstractly, the answer is yes, but such convex hulls may be hard to characterize, and the corresponding penalty may not be efficiently computable.
It is interesting to find special cases where this construction can be carried out and results in a tractable problem. Recent developments in this direction include the ``square norm'' proposed by \cite{SquareDeal} for the low-rank tensor recovery; which provably outperforms \eqref{SSN} for Gaussian measurements and the $(k,q)$-trace norm introduced by Richard et al. to estimate S\&L matrices \cite{Richard}.

\paragraph{Algorithms for minimizing combination of norms.}
Despite the limitation in their theoretical performance, in practice one may still need to solve convex relaxations that combine the different norms, i.e., problem \eqref{rec_class}. Consider the special case of sparse and low-rank matrix recovery. All corresponding optimization problems mentioned in Theorem \ref{main_thm_sec3} can be expressed as a semidefinite program and solved by standard solvers; for example, for the numerical experiments in Section \ref{sec:numerical} we used the interior-point solver SeDuMi \cite{sedumi} via the modeling environment CVX \cite{cvx}. However, interior point methods do not scale for problems with tens of thousands of matrix entries, which are common in machine learning applications. One future research direction is to explore first-order methods, which have been successful in solving problems with a single structure (for example $\ell_1$ or nuclear norm regularization alone). In particular, Alternating Directions Methods of Multipliers (ADMM) appears to be a promising candidate.

\paragraph{Characterizing the tightness of the lower bounds.}
The results provided in this paper are negative in nature, as we characterize the lower bounds on the required amount of measurements for mixed convex recovery problems. However, it would be interesting to see how much we can gain by making use of multiple norms and how tight are these lower bounds. In \cite{samet}, authors investigate a specific simultaneous model where signal $\x\in\R^n$ is sparse in both time and frequency domains, i.e., $\x$ and $\Db\x$ are $k_1,k_2$ sparse respectively where $\Db$ is the Discrete Fourier Transform matrix. For recovery, the authors consider minimizing $\|\x\|_1+\la\|\Db\x\|_1$ subject to measurements. Intuitively, results of this paper would suggest the necessity of $\Omega(\min\{k_1,k_2\})$ measurements for successful recovery. On the other hand, best of the individual functions ($\ell_1$ norms) will require $\Omega(\min\{k_1\log\frac{n}{k_1},k_2\log \frac{n}{k_2}\})$ measurements. In \cite{samet}, it is shown that the mixed approach will require as little as $\max\{k_1,k_2\}\log\log n$ under mild assumptions. 

This shows that the mixed approach can result in a logarithmic improvement over the individual functions when $k_1\approx k_2$ and the lower bound given by this paper might be achievable up to a small factor.

\paragraph{Connection to Sparse PCA.}
The sparse PCA problem (see, e.g. \cite{NesterovSPCA,PCA1,d'Aspremontetal}) seeks sparse principal components given a (possibly noisy) data matrix. Several formulations for this problem exist, and many algorithms have been proposed. In particular, a popular algorithm is the SDP relaxation proposed in \cite{d'Aspremontetal}, which is based on the following formulation.

For the first principal component to be sparse, we seek an $\x\in\R^{n}$ that maximizes $\x^T\mathbf{A}\x$ for a given data matrix $\mathbf{A}$, and minimizes $\|\x\|_0$. Similar to the sparse phase retrieval problem, this problem can be reformulated in terms of a rank-1, PSD matrix $\X=\x\x^T$ which is also row- and column-sparse.
Thus we seek a simultaneously low-rank and sparse $\X$. This problem is different from the recovery problem studied in this paper, since we do not have $m$ random measurements of
$\X$. Yet, it will be interesting to connect this paper's results to the sparse PCA problem to potentially provide new insights for sparse PCA.

\paragraph{\textbf{Acknowledgements.}}
This work was supported in part by the National Science Foundation under grants CCF-0729203, CNS-0932428 and CCF-1018927, by the Office of Naval Research under the MURI grant N00014-08-1-0747, by Caltech's Lee Center for Advanced Networking,
and by the National Science Foundation CAREER award ECCS-0847077.
The work of Y. Eldar is supported in part by the Israel Science Foundation under Grant no. 170/10, in part by the Ollendorf Foundation, and in part by a Magnet grant Metro450 from the Israel Ministry of Industry and Trade.


\begin{thebibliography}{9}


\bibitem{candes-tao}
E.J. Cand\`es and T. Tao, ``Decoding by linear programming," IEEE Trans. Inform. Theory, 51 4203-4215.

\bibitem{donoho}
D.L. Donoho, ``Compressed sensing," IEEE Trans. Inform. Theory, 52(4):1289-1306, 2006.
\bibitem{candes-tao2}
E.J. Candes, JK Romberg and T Tao, ``Stable signal recovery from incomplete and inaccurate measurements''. Comm. on Pure and Applied Math. Vol. 59, Issue 8, pg 1207Ð1223, August 2006.

\bibitem{RFP}
B. Recht, M. Fazel, P. Parrilo, ``Guaranteed Minimum-Rank Solutions of Linear Matrix Equations via Nuclear Norm Minimization''. SIAM Review, Vol 52, no 3, pages 471-501, 2010.

\bibitem{CandesRecht-completion}
E.J. Cand\`es and B. Recht, ``Exact matrix completion via convex optimization," Found. of Comput. Math., 9 717-772.

\bibitem{ChandraParriloWillsky-SL}
V. Chandrasekaran, P. A. Parrilo, and A. S. Willsky, ``Latent Variable Graphical Model Selection via Convex Optimization", Annals of Statistics.

\bibitem{candes_wright}
E.J. Cand\`es, X. Li, Y. Ma, J. Wright, ``Robust Principal Component Analysis?''. Journal of ACM 58(1), 1-37.

\bibitem{chandra}
V. Chandrasekaran, B. Recht, P. A. Parrilo, A. S. Willsky, ``The Convex Geometry of Linear Inverse Problems''. arXiv:1012.0621v3.

\bibitem{ALMT}
D. Amelunxen, M. Lotz, M. B. McCoy, and J. A. Tropp. ``Living on the edge: Phase transitions in convex programs with random data." Inform. Inference (2014).

\bibitem{MF13}
R. Foygel and L. Mackey. ``Corrupted sensing: Novel guarantees for separating structured signals,'' Information Theory, IEEE Transactions on 60.2 (2014): 1223--1247.

\bibitem{DonTan}
D. L. Donoho and J. Tanner. ``Sparse nonnegative solution of underdetermined linear equations by linear programming.'' Proceedings of the National Academy of Sciences of the United States of America 102.27 (2005): 9446-9451.
\bibitem{OymLas}
S. Oymak, C. Thrampoulidis, and B. Hassibi. ``The squared-error of generalized lasso: A precise analysis.'' arXiv:1311.0830.

\bibitem{VerEst}
R. Vershynin, ``Estimation in high dimensions: a geometric perspective'', arXiv:1405.5103.

\bibitem{FusedLasso}
R. Tibshirani, M. Saunders, S. Rosset, J. Zhu, and K. Knight. ``Sparsity and smoothness via the fused lasso.'' Journal of the Royal Statistical Society: Series B (Statistical Methodology) 67, no. 1 (2005): 91-108.
\bibitem{Ward}
D. Needell and R. Ward. ``Stable image reconstruction using total variation minimization." SIAM Journal on Imaging Sciences 6.2 (2013): 1035-1058.

\bibitem{SquareDeal}
C. Mu, B. Huang, J. Wright, and D. Goldfarb, ``Square deal: Lower bounds and improved relaxations for tensor recovery,'' arXiv:1307.5870.


\bibitem{Tucker}
 L. Tucker. Some mathematical notes on three-mode factor analysis.
Psychometrika, 31(3):279-311, 1966.
\bibitem{RechtTensor}
 S. Gandy, B. Recht, and I. Yamada, ``Tensor completion and low-n-rank tensor recovery via convex optimization'', Inverse Problems 27(2), 025010 (2011).
\bibitem{TenSur}
Grasedyck, Lars, Daniel Kressner, and Christine Tobler. ``A literature survey of low?rank tensor approximation techniques.'' GAMM?Mitteilungen 36.1 (2013): 53--78.
\bibitem{TenVis}
Liu, Ji, et al. ``Tensor completion for estimating missing values in visual data." Pattern Analysis and Machine Intelligence, IEEE Transactions on 35.1 (2013): 208-220.
\bibitem{TVNuc1} Semerci, O., Hao, N., Kilmer, M. E., and Miller, E. L. (2013). Tensor-based formulation and nuclear norm regularization for multi-energy computed tomography. arXiv preprint arXiv:1307.5348.
\bibitem{TVNuc2} Golbabaee, M., and Vandergheynst, P. (2012, September). Joint trace/TV norm minimization: A new efficient approach for spectral compressive imaging. In Image Processing (ICIP), 2012 19th IEEE International Conference on (pp. 933-936). IEEE.


\bibitem{PlanRIPless}
E. J. Candes, and Y. Plan. ``A probabilistic and RIPless theory of compressed sensing.'' Information Theory, IEEE Transactions on 57.11 (2011): 7235-7254.

\bibitem{AgarwalNegahbanWainwright2012}
A. Agarwal, S. Negahban, M. J. Wainwright, ``Noisy matrix decomposition via convex relaxation: Optimal rates in high dimensions," Annals of Statistics, Volume 40, Number 2 (2012), 1171-1197.

\bibitem{CandesPlan}
E.J. Cand\`es, Y. Plan. ``Tight oracle bounds for low-rank matrix recovery from a minimal number of random measurements," IEEE Transactions on Information Theory 57(4), 2342-2359.

\bibitem{candes_tao}
E.J. Cand\`es, J. Romberg, T. Tao, ``Robust uncertainty principles: exact signal reconstruction from highly incomplete frequency information''. IEEE Trans. Inform. Theory, 52 489--509.

\bibitem{wright12}
J. Wright, A. Ganesh, K. Min, Y. Ma, ``Compressive Principal Component Pursuit''. arXiv:1202.4596v1.

\bibitem{StojnicBlock}
M. Stojnic, F. Parvaresh, and B. Hassibi, ÒOn the reconstruction of block-sparse signals with an optimal number of measurements,Ó IEEE Trans. Signal Process., vol. 57, no. 8, pp. 3075Ð3085, May 2010.


\bibitem{YuanL12}
M. Yuan and Y. Lin, ÒModel selection and estimation in regression with grouped variables,Ó J. Roy. Stat. Soc. Ser. B Stat. Methodol., vol. 68, no. 1, pp. 49Ð67, 2006.

\bibitem{SRSE11}
P. Sprechmann, I. Ramirez, G. Sapiro, Y.C. Eldar, ``C-HiLasso: A Collaborative Hierarchical Sparse Modeling Framework", IEEE Transactions on Signal Processing, vol.59, issue 9, pp.4183-4198, Sept. 2011.

\bibitem{vander}
M. Golbabaee and P. Vandergheynst, ``Hyperspectral image compressed sensing via low-rank and joint-sparse matrix recovery''. ICASSP 2012.

\bibitem{SESS11} 
Y. Shechtman and Y.C. Eldar and A. Szameit and M. Segev,
``Sparsity-based sub-wavelength imaging with partially spatially incoherent light via quadratic compressed sensing'', Optics Express, 19:14807--14822 , 2011.

\bibitem{BE12}A. Beck, Y.C. Eldar, ``Sparsity constrained nonlinear optimization: Optimality conditions and algorithms'', arXiv:1203.4580v1,

\bibitem{SEA12}A. Szameit et. al., ``Sparsity-based single-shot sub-wavelength coherent diffractive imaging'', Nature Materials.

\bibitem{W63}A. Walther. ``The question of phase retrieval in
optics''. Opt. Acta, 10:41\textendash{}49, 1963.
\bibitem{millane}
R.P. Millane, ``Phase retrieval in crystallography and optics''. J. Opt. Soc. Am. A 7, 394-411 (1990).

\bibitem{H93}R.W. Harrison, ``Phase problem in crystallography''.
J. Opt. Soc. Am. A, 10(5):1045\textendash{}1055, 1993.


\bibitem{H89}N. Hurt. ``Phase retrieval and zero crossings'',
Kluwer Academic Publishers, Norwell, MA, 1989.


\bibitem{OYDS12}H. Ohlsson, A. Y. Yang, R. Dong, S. S. Sastry, ``Compressive
phase retrieval from squared output measurements via semidefinite
programming'', arXiv:1111.6323v3 , March 2012.




\bibitem{LiVoroninski2012}
X. Li, V. Voroninski, ``Sparse Signal Recovery from Quadratic Measurements via Convex Programming," arXiv:1209.4785.


\bibitem{JOB12}
K. Jaganathan, S. Oymak, and B. Hassibi, ``Recovery
of sparse 1-D signals from the magnitudes of their Fourier transform'',
arXiv:1206.1405v1, June 2012.

\bibitem{Vetterli}
Y.M. Lu and M. Vetterli, ``Sparse Spectral Factorization: Unicity and Reconstruction Algorithms''. Acoustics, Speech and Signal Processing (ICASSP), 2011 IEEE International Conference on, pp. 5976-5979, 22-27 May 2011.


\bibitem{SBE12} Y. Shechtman and A. Beck and Y.C. Eldar, ``Efficient Phase Retrieval of Sparse Signals''. IEEI 2012.
\bibitem{EldarRobust}
Y.C. Eldar and S. Mendelson "Phase Retrieval: Stability and Recovery Guarantees", Arxiv:1211.0872, Nov. 2012.

\bibitem{CESV12}E.J. Cand\`es, Y.C. Eldar, T. Strohmer and V. Voroninski,
``Phase retrieval via matrix completion'', arXiv:1109.0573, Sep. 2011.
\bibitem{candes_strohmer}
 E.J. Cand\`es, T. Strohmer, V. Voroninski. ``PhaseLift: exact and stable signal recovery from magnitude measurements via convex programming''. to appear in Communications on Pure and Applied Mathematics.



\bibitem{F82}J.R. Fienup, ``Phase retrieval algorithms: a comparison'',
Applied Optics 21, 2758-2769, 1982.

\bibitem{GS72}R.W. Gerchberg and W.O. Saxton, ``Phase retrieval by
iterated projections'', Optik 35, 237, 1972.


\bibitem{candes_recht12}
E.J. Cand\`es and B. Recht, ``Simple Bounds for Recovering Low-complexity Models''. arXiv:1106.1474v2.

\bibitem{TroppDict}
J. A. Tropp ``On the conditioning of random subdictionaries." Applied and Computational Harmonic Analysis 25.1 (2008): 1--24.

\bibitem{watson}
G.A. Watson, ``Characterization of the subdifferential of some matrix norms''. Linear Algebra and its Appl.
170, 33--45 (1992).

\bibitem{boyd}
S. Boyd and L. Vandenberghe, ``Convex Optimization''. Cambridge University Press, 2004.
\bibitem{Bertse}
D. Bertsekas with A. Nedic and A.E. Ozdaglar, ``Convex Analysis and Optimization'' Athena Scientific, 2003.
\bibitem{rock_var}
R.T. Rockafellar, R. J-B Wets, ``Variational Analysis". Springer, 2004.

\bibitem{rock}
R.T. Rockafellar, ``Convex Analysis". Princeton University Press, 1997.

\bibitem{cvx}
CVX Research, Inc. ``CVX: Matlab software for disciplined convex programming'', version 2.0 beta. http://cvxr.com/cvx, September 2012.

\bibitem{sedumi}
J.F. Sturm, ``Using SeDuMi 1.02, a MATLAB toolbox for optimization over symmetric cones,'' 1998.

\bibitem{Richard}
E. Richard, G. Obozinski, J.-P. Vert, ``Tight convex relaxations for sparse matrix factorization'', arXiv:1407.5158.


\bibitem{samet}
S. Oymak and B. Hassibi, ``Recovering Jointly Sparse Signals via Joint Basis Pursuit''. arXiv:1202.3531v1.


\bibitem{NesterovSPCA}
M. Journ\'ee, Y. Nesterov, P. Richt\'arik, and R Sepulchre, ``Generalized Power Method for Sparse Principal Component Analysis''. Journal of Machine Learning Research 11 (2010) 517--553.


\bibitem{PCA1}
H. Zou, T. Hastie, and R. Tibshirani, ``Sparse Principal Component Analysis'', Journal of computational and graphical statistics 15 (2), 265-286.

\bibitem{d'Aspremontetal}
A. d'Aspremont, F. Bach, L. El Ghaoui. ``Optimal Solutions for Sparse Principal Component Analysis'', Journal of Machine Learning Research, 9(Jul):1269Ð1294, 2008.

\bibitem{moreau}
J.J. Moreau, ``D\'ecomposition orthogonale d'un espace Hilbertien selon deux cones mutuellement polaires''. C. R. Acad. Sci., volume 255, pages 238--240, 1962.

\bibitem{Gordon}
Y. Gordon, ``On Milman's inequality and random subspaces which escape through a mesh in $\R^n$''. Geometric aspects of functional analysis, Isr. Semin. 1986-87, Lect. Notes Math. 1317, 84-106. (1988).



\bibitem{Bai-Yin}
Z. Bai and Y. Yin, ``Limit of the smallest eigenvalue of a large-dimensional sample covariance matrix'', Annals of Probability, 21, 1275--1294 (1993).
\bibitem{Vershynin}
R. Vershynin, ``Introduction to the non-asymptotic analysis of random matrices''. arXiv:1011.3027v7.





















 



\end{thebibliography}

\newpage
\appendix
\begin{center}\large{{APPENDIX}}\end{center}
\section{Properties of Cones} \label{app_cone}
In this appendix, we state some results regarding cones which are used in the proof of general recovery. Recall the definitions of polar and dual cones from Section \ref{setup}.

\begin{theorem}[Moreau's decomposition theorem, \cite{moreau}] \label{decamp}Let $\Cc$ be a closed and convex cone in $\R^n$. Then, for any $\x\in\R^n$, we have
\begin{itemize}
\item $\x=\Pc_\Cc(\x)+\Pc_\Cp(\x)$.
\item $\li\Pc_\Cc(\x),\Pc_\Cp(\x)\ri=0$.
\end{itemize}
\end{theorem}

\begin{lemma}[Projection is nonexpansive]\label{lem_proj}
Let $\Cc\in\R^n$ be a closed and convex set and $\ab,\bb\in\R^n$ be vectors. Then,
\beq
\|\Pc_\Cc(\ab)-\Pc_\Cc(\bb)\|_2\leq\|\ab-\bb\|_2.\nn
\eeq
\end{lemma}
\begin{corollary}\label{wanter} Let $\Cc$ be a closed convex cone and $\ab,\bb$ be vectors satisfying $\Pc_\Cc(\ab-\bb)=0$. Then
\beq
\|\bb\|_2\geq \|\Pc_\Cc(\ab)\|_2.\nn
\eeq
\end{corollary}
\begin{proof}
Using Lemma \ref{lem_proj}, we have $\|\Pc_\Cc(\ab)\|_2=\|\Pc_\Cc(\ab)-\Pc_\Cc(\ab-\bb)\|_2\leq \|\bb\|_2$.
\end{proof}
%
The unit sphere in $\R^n$ will be denoted by $\Sc^{n-1}$ for the following theorems.
\begin{theorem}[Escape through a mesh, \cite{Gordon}]\label{ETM}
For a given set $\D\in\Sc^{n-1}$, define the \emph{Gaussian width} as
\[
\omega(\D) = \E \left[ \sup_{\x\in\D} \iprod{\x}{\g} \right],
\]
in which $\g\in\R^n$ has i.i.d.~standard Gaussian entries. Given $m$, let $d=\sqrt{n-m}-\frac{1}{4\sqrt{n-m}}$. Provided that $\omega(\D)\leq d$ a random $m-$dimensional subspace which is uniformly drawn w.r.t.~Haar measure will have no intersection with $\D$ with probability at least
\beqa \label{gordon_prob}
1-3.5 \exp(-(d-\omega(\D))^2).
\eeqa
\end{theorem}

\begin{theorem} \label{proj_norm}
Consider a random Gaussian map $\Gc:\R^n\to\R^m$ with i.i.d.~entires and the corresponding adjoint operator $\Gc^*$. Let $\Cc$ be a closed and convex cone and recalling Definition \ref{cwidth}, let
\[
\zeta(\Cc):=1-\bDb(\Cc),~\gamma(\Cc):= 2\sqrt{\frac{1+\bDb(\Cc)}{1-\bDb(\Cc)}}.
\]
where $\bDb(\Cc)=\frac{\Db(\Cc)}{\sqrt{n}}$. Then, if $m\leq \frac{7\zeta(\Cc)}{16}n$, with probability at least $1-6\exp(-(\frac{\zeta(\Cc)}{4})^2n)$, for all $\z\in\R^n$ we have
\beqa \label{eq:proj_norm}
\twonorm{\Gc^*(\z)} \leq \gamma(\Cc) \twonorm{\Pc_\Cc (\Gc^*(\z))}.
\eeqa
\end{theorem}
\begin{proof}
For notational simplicity, let $\zeta=\zeta(\Cc)$ and $\gamma=\gamma(\Cc)$. Consider the set
\[
\D = \left\{  \x\in\Sc^{n-1}: \;\; \twonorm{\x} \geq \gamma \twonorm{\Pc_\Cc(\x)}  \right\}.
\]
and we are going to show that with high probability, the range of $\Gc^*$ misses $\D$. 
Using Theorem \ref{decamp}, for any $\x\in\D$, we may write
\begin{align}
\iprod{\x}{\g} &=\iprod{\Pc_\Cc(\x)+\Pc_\Cp(\x)}{\Pc_\Cc(\g)+\Pc_\Cp(\g)}\nn\\
&\leq \iprod{\Pc_\Cc(\x)}{\Pc_\Cc(\g)}  + \iprod{\Pc_\Cp(\x)}{\Pc_\Cp(\g)} \label{ineq1}\\
&\leq \twonorm{\Pc_\Cc(\x)}\twonorm{\Pc_\Cc(\g)}  + \twonorm{\Pc_\Cp(\x)}\twonorm{\Pc_\Cp(\g)} \nn\\
&\leq \gamma^{-1}\twonorm{\Pc_\Cc(\g)}  + \twonorm{\Pc_\Cp(\g)}\nn
\end{align}
where in \eqref{ineq1} we used the fact that elements of $\Cc$ and $\Cp$ have nonpositive inner products and $\twonorm{\Pc_\Cc(\x)}  \leq \twonorm{\x}$ is by Lemma \ref{lem_proj}.
Hence, from the definition of Gaussian width,
\begin{align*}
\omega(\D) =  \E \left[ \sup_{\x \in \D} \iprod{\x}{\g} \right] &\leq  \gamma^{-1}\E \left[  \twonorm{\Pc_\Cc(\g)}  \right] + \E \left[ \twonorm{\Pc_\Cp(\g)}  \right] \\
&\leq \sqrt{n}(\gamma^{-1} \bDb(\Cp)+\bDb(\Cc)) \leq \frac{2-\zeta}{2}\sqrt{n}.
\end{align*}
Where we used the fact that $\gamma\geq \frac{2\bDb(\Cp)}{1-\bDb(\Cc)}$; which follows from $\bDb(\Cc)^2+\bDb(\Cp)^2\leq 1$ (see Theorem \ref{decamp} above). Hence, whenever,
\beq
m\leq \frac{7\zeta}{16}n\leq (1-(\frac{4-\zeta}{4})^2)n = m',\nn
\eeq
using the upper bound on $\omega(\D)$, we have,
\begin{align}
(\sqrt{n-m}-\omega(\D)-\frac{1}{4\sqrt{n-m}})^2\geq (\sqrt{n-m}-\omega(\D))^2-\frac{1}{2}\geq (\frac{\zeta}{4})^2n-\frac{1}{2}.
\end{align}
Now, using Theorem \ref{ETM}, the range space of $\Gc^*$ will miss the undesired set $\D$ with probability at least $ 1-3.5\exp(-(\frac{\zeta}{4})^2n+\frac{1}{2})\geq 1-6\exp(-(\frac{\zeta}{4})^2n)$.
\end{proof}
%
%
%
%
%
%
\begin{lemma} \label{psdcorol}Consider the cones $\Sbb^\md$ and $\Sbb^\md_+$ in the space $\R^{\md\times \md}$. Then, $\bDb(\Sbb^\md)<\frac{1}{\sqrt{2}}$ and $\bDb(\Sbb^\md_+)<\frac{\sqrt{3}}{2}$.

\end{lemma}

\begin{proof} Let $\G$ be a $\md\times \md$ matrix with i.i.d. standard normal entries. Set of symmetric matrices $\Sbb^\md$ is an $\frac{\md(\md+1)}{2}$ dimensional subspace of $\R^{\md\times \md}$. Hence, $\E\|\Pc_{\Sbb^\md}(\G)\|_F^2=\frac{\md(\md+1)}{2}$ and $\E\|\Pc_{(\Sbb^\md)^\circ}(\G)\|_F^2=\frac{\md(\md-1)}{2}$. Hence,
\beq
\bDb(\Sbb^\md)=\sqrt{\frac{\md(\md-1)}{2\md^2}}<\frac{1}{\sqrt{2}}.\nn
\eeq
To prove the second statement, observe that projection of a matrix $\Ab\in\R^{\md\times \md}$ onto $\Sbb^\md_+$ is obtained by first projecting $\Ab$ onto $\Sbb^\md$ and then taking the matrix induced by  the positive eigenvalues of $\Pc_{\Sbb^\md}(\Ab)$. Since, $\G$ and $-\G$ are identically distributed and $\Sbb^\md_+$ is a self dual cone, $\Pc_{\Sbb^\md_+}(\G)$ is identically distributed as $-\Pc_{\Sbb^\md_-}(\G)$ where $\Sbb^\md_- = (\Sbb^\md_+)^\circ$ stands for negative semidefinite matrices. Hence,
\beq
\E\|\Pc_{\Sbb^\md_+}(\G)\|_F^2=\frac{\E\|\Pc_{\Sbb^\md}(\G)\|_F^2}{2}=\frac{\md(\md+1)}{4},~\E\|\Pc_{(\Sbb^\md_+)^\circ}(\G)\|_F^2=\frac{\md(3\md-1)}{4}.\nn
\eeq
Consequently, $\bDb(\Sbb^\md_+)=\sqrt{\frac{3}{4}-\frac{1}{4\md}}<\sqrt{\frac{3}{4}}$.

\end{proof}

\section{Norms in Sparse and Low-rank Model} \label{app:norms}   \label{sec:proof51}

\subsection{Relevant notation for the proofs}\label{appnot}
Let $[k]$ denote the set $\{1,2,\dots,k\}$. Let $S_c,S_r$ denote the indexes of the nonzero columns and rows of $\X_0$ so that nonzero entries of $\X_0$ lies on $S_r\times S_c$ submatrix. $\Sc_c,\Sc_r$ denotes the $k_1,k_2$ dimensional subspaces of vectors whose nonzero entries lie on $S_c$ and $S_r$ respectively.

Let $\X_0$ have singular value decomposition $\U\bSi\V^T$ such that $\bSi\in \R^{r\times r}$ and columns of $\U,\V$ lies on $\Sc_c,\Sc_r$ respectively.

\subsection{Proof of Lemma \ref{signvec}}
\begin{proof}
Observe that $T_c=\R^\md\times \Sc_c$ and $T_r=\Sc_r\times \R^\md$ hence $T_c\cap T_r$ is the set of matrices that lie on $S_r\times S_c$. Hence, $\Eb_\st=\U\V^T\in T_c\cap T_r$. Similarly, $\Eb_c$ and $\Eb_r$ are the matrices obtained by scaling columns and rows of $\X_0$ to have unit size. As a result, they also lie on $S_r\times S_c$ and $T_c\cap T_r$. $\Eb_\st\in T_\st$ by definition.

Next, we may write $\Eb_c = \X_0 \Db_c$ where $\Db_c$ is the scaling nonnegative diagonal matrix. Consequently, $\Eb_c$ lies on the range space of $\X_0$ and belongs to $T_\st$. This follows from definition of $T_\st$ in Lemma \ref{lemdecomp} and the fact that $(\Ib-\U\U^T)\Eb_c=0$.

In the exact same way, $\Eb_r=\Db_r\X_0$ for some nonnegative diagonal $\Db_r$ and lies on the range space of $\X^T$ and hence lies on $T_\st$. Consequently, $\Eb_\st,\Eb_c,\Eb_r$ lies on $T_c\cap T_r\cap T_\st$.

Now, consider
\[
\iprod{\Eb_c}{\Eb_\st} = \iprod{\X_0\Db_c}{\U\V^T} = \tr{\V\U^T\U\bSi\V^T\Db_c} = \tr{\V\bSi\V^T\Db^c}\geq 0.
\]
since both $\V\bSi\V^T$ and $\Db^c$ are positive semidefinite matrices. In the exact same way, we have $\iprod{\Eb_c}{\Eb_\st}\geq 0$. Finally,
\beq
\iprod{\Eb_c}{\Eb_r}= \iprod{\X_0\Db_c}{\Db_r\X_0}   =\tr{\Db_c\X_0^T\Db_r\X_0} \geq 0,\nn
\eeq
since both $\Db_c$ and $\X_0^T\Db_r\X_0$ are PSD matrices. Overall, the pairwise inner products of $\Eb_r,\Eb_c,\Eb_\st$ are nonnegative.


\end{proof}

\subsection{Results on the positive semidefinite constraint}
\begin{lemma}\label{verysimp} Assume $\X,\Y\in \Sbb_+^\md$ have eigenvalue decompositions $\X=\sum_{i=1}^{\text{rank}(\X)}\sigma_i\ub_i\ub_i^T$ and $\Y=\sum_{i=1}^{\text{rank}(\Y)} c_i\vb_i\vb_i^T$. Further, assume $\li\Y,\X\ri=0$. Then, $\U^T\Y=0$ where $\U=[\ub_1~\ub_2~\dots~\ub_{\text{rank}(\X)}]$.
\end{lemma}
\begin{proof} Observe that,
\beq
\li\Y,\X\ri=\sum_{i=1}^{\text{rank}(\X)}\sum_{j=1}^{\text{rank}(\Y)}\sigma_ic_j|\ub_i^T\vb_j|^2.\nn
\eeq
Since $\sigma_i,c_j>0$, right hand side is $0$ if and only if $\ub_i^T\vb_j=0$ for all $i,j$. Hence, the result follows.
\end{proof}

\begin{lemma} \label{useful psd}Assume $\X_0\in\Sbb_+^\md$ so that in Section \ref{appnot}, $S_c=S_r$, $T_c=T_r$, $k_1=k_2=k$ and $\U=\V$. Let $\Rc= T_c\cap T_r\cap T_\st \cap \Sbb^\md$, $S_\st=T_\st\cap \Sbb^\md$, and,
\beq
\Yc=\{\Y\big|\Y\in (\Sbb_+^\md)^*,~\li\Y,\X_0\ri=0\},\nn
\eeq
 Then, the following statements hold.
\begin{itemize}
\item $S_\st\subseteq\text{span}(\Yc)^\perp$. Hence, $\tcs\subseteq S_\st$ and is orthogonal to $\Yc$.
\item $\Eb_\st\in \tcs$, $\frac{\|\Pc_{\tcs}(\Eb_c)\|_F}{\|\Eb_c\|_F}=\frac{\|\Pc_{\tcs}(\Eb_r)\|_F}{\|\Eb_r\|_F}\geq \frac{1}{\sqrt{2}}$.
\end{itemize}
\end{lemma}
\begin{proof}
The dual of $\Sbb_+^\md$ with respect to $\R^{\md\times \md}$ is the set sum of $\Sbb_+^\md$ and $\text{Skew}^\md$ where $\text{Skew}^\md$ is the set of skew-symmetric matrices. Now, assume, $\Y\in \Yc$ and $\X\in S_\st$. Then, $\li\Y,\X\ri=\li\frac{\Z}{2},\X\ri$ where $\Z=\Y+\Y^T\in \Sbb_+^\md$ and $\li\Z,\X_0\ri=0$. Since $\X_0$, $\Z$ are both PSD, applying Lemma \ref{verysimp}, we have $\U^T\Z=0$ hence $(\Ib-\U\U^T)\Z(\Ib-\U\U^T)=\Z$ which means $\Z\in T_\st^\perp$. Hence, $\li\Z,\X\ri=\li\Y,\X\ri=0$ as $\X\in S_\st\subset T_\st$. Hence, $\text{span}(\Yc)\subseteq S_\st^\perp$. 

%

For the second statement, let $T_\cap=T_\st\cap T_c\cap T_r$. Recalling Lemma \ref{signvec}, observe that $\Eb_\st\in T_\cap$. Since $\Eb_\st$ is also symmetric, $\Eb_\st\in\tcs$. Similarly, $\Eb_c,\Eb_r\in T_\cap$, $\li\Eb_c,\Eb_r\ri\geq 0$ and $\|\Pc_\tcs(\Eb_c)\|=\|\frac{\Eb_c+\Eb_r}{2}\|_F\geq \frac{\|\Eb_c\|_F}{\sqrt{2}}$. Similar result is true for $\Eb_r$.

%
\end{proof}

\bigskip
\section{Results on non-convex recovery}\label{appC}

Next two lemmas are standard results on sub-gaussian measurement operators.
\begin{lemma}[Properties of sub-gaussian mappings]\label{ranemb} Assume $\X$ is an arbitrary matrix with unit Frobenius norm. A measurement operator $\Mec(\cdot)$ with i.i.d zero-mean isotropic subgaussian rows (see Section \ref{measure}) satisfies the following:
\begin{itemize}
\item $\E[\|\Mec(\X)\|_2^ 2]=m$.
\item There exists an absolute constant $c>0$ such that, for all $1\geq \eps\geq 0$, we have
\beq
\Pro(|\|\Mec(\X)\|_2^2-m|\geq \eps m)\leq 2\exp(-c\eps^2m).\nn
\eeq
\end{itemize}
\end{lemma}
\begin{proof}Observe that, when $\|\X\|_F=1$, entries of $\Mec(\X)$ are zero-mean with unit variance. Hence, the first statement follows directly. For the second statement, we use the fact that square of a sub-gaussian random variable is sub-exponential and view $\|\Mec(\X)\|_2^2$ as a sum of $m$ i.i.d. subexponentials with unit mean. Then, result follows from Corollary 5.17 of \cite{Vershynin}.
\end{proof}


For the consequent lemmas, $\Sc^{\md_1\times \md_2}$ denotes the unit Frobenius norm sphere in $\R^{\md_1\times \md_2}$.
\begin{lemma}\label{ranemb2}
Let $\Dc\in\R^{\md_1\times \md_2}$ be an arbitrary cone and $\Mec(\cdot):\R^{\md_1\times \md_2}\rightarrow \R^m$ be a measurement operator with i.i.d zero-mean and isotropic sub-gaussian rows. Assume that the set $\bar\Dc = \Sc^{\md_1\times \md_2}\cap \Dc$ has $\eps$-covering number bounded above by $\eta(\eps)$. Then, there exists constants $c_1,c_2>0$ such that whenever $m\geq c_1\log \eta(1/4)$, with probability $1-2\exp(-c_2 m)$, we have
\beqas
 \Dc\cap \Null(\Mec)=\{0\}.\nn
\eeqas
\end{lemma}

\begin{proof} Let $\eta=\eta(\frac{1}{4})$, and $\{\X_i\}_{i=1}^{\eta}$ be a $\frac{1}{4}$-covering of $\bar\Dc$. With probability at least $1-2\eta\exp(-c\eps^2m)$, for all $i$, we have
\beq
(1-\eps)m\leq \|\Ac(\X_i)\|_2^2\leq (1+\eps)m.\nn
\eeq
Now, let $\X_{\sup}=\arg\sup_{\X\in\bar\Dc}\|\A(\X)\|_2$. Choose $1\leq a\leq \eta$ such that $\|\X_a-\X_{\sup}\|_2\leq 1/4$. Then:
\beq
\|\A(\X_{\sup})\|_2\leq \|\A(\X_a)\|_2+\|\A(\X_{\sup}-\X_a)\|_2\leq (1+\eps)m+\frac{1}{4} \|\A(\X_{\sup})\|_2.\nn
\eeq
Hence, $\|\A(\X_{\sup})\|_2\leq \frac{4}{3}(1+\eps)m$. Similarly, let $\X_{\inf}=\arg\inf_{\X\in\bar\Dc}\|\A(\X)\|_2$. Choose $1\leq b\leq \eta$ satisfying $\|\X_b-\X_{\inf}\|\leq 1/4$. Then,
\beq
\|\A(\X_{\inf})\|_2\geq \|\A(\X_b)\|_2-\|\A(\X_{\inf}-\X_b)\|_2\geq (1-\eps)m-\frac{1}{3}(1+\eps)m.\nn
\eeq
This yields $\|\A(\X_{\inf})\|_2\geq \frac{2-4\eps}{3}m$. Choosing $\eps=1/4$ whenever $m\geq \frac{32}{c}\log(\eta)$ with the desired probability, $\|\A(\X_{\inf})\|_2>0$. Equivalently, $\bar\Dc\cap \Null(\A)=\emptyset$. Since $\A(\cdot)$ is linear and $\Dc$ is a cone, the claim is proved.
\end{proof}

The following lemma gives a covering number of the set of low rank matrices.
\begin{lemma}[Candes and Plan, \cite{CandesPlan}] \label{covnum}Let $M$ be the set of matrices in $\R^{\md_1\times \md_2}$ with rank at most $r$. Then, for any $\eps>0$, there exists a covering of $\Sc^{\md_1\times \md_2}\cap M$ with size at most $(\frac{c_3}{\eps})^{(\md_1+\md_2)r}$ where $c_3$ is an absolute constant. In particular, $\log(\eta(1/4))$ is upper bounded by $C^{(\md_1+\md_2)r}$ for some constant $C>0$.
\end{lemma}
Now, we use Lemma \ref{covnum} to find the covering number of the set of simultaneously low rank and sparse matrices.
\subsection{Proof of Lemma \ref{abcgcc}}
\begin{proof} Assume $M$ has $\frac{1}{4}$-covering number $N$. Then, using Lemma \ref{ranemb2}, whenever $m\geq c_1\log N$, (\ref{abcgcceq}) will hold. What remains is to find $N$. To do this, we cover each individual $s_1\times s_2$ submatrix and then take the union of the covers. For a fixed submatrix, using Lemma \ref{covnum}, $\frac{1}{4}$-covering number is given by $C^{(s_1+s_2)q}$. In total there are ${\md_1\choose s_1}\times {\md_2\choose s_2}$ distinct submatrices. Consequently, by using $\log {\md\choose s}\approx s\log \frac{\md}{s}+s$, we find
\beq
\log N\leq \log\left({\md_1\choose s_1}\times {\md_2\choose s_2}C^{(s_1+s_2)q}\right)\leq s_1\log\frac{\md_1}{s_1}+s_1+s_2\log\frac{\md_2}{s_2}+s_2+(s_1+s_2)q\log C,\nn
\eeq
and obtain the desired result.
\end{proof}

%

\end{document}